\documentclass[11pt]{article}

\usepackage{amssymb,amsmath}
\usepackage{fullpage}
\usepackage{amsthm}
\usepackage{array}

\usepackage[usenames,dvipsnames]{color}
\usepackage[colorlinks=true,linkcolor=blue,citecolor=Mahogany]{hyperref}
\usepackage{enumitem}

\usepackage{eulervm}
\usepackage{graphicx}

\usepackage{amsopn}
\usepackage{mathtools}

\usepackage{charter}

\renewcommand{\leq}{\leqslant}

\renewcommand{\ge}{\geqslant}
\renewcommand{\le}{\leqslant}

\usepackage{nicefrac}

\newcommand{\nfrac}{\nicefrac}
\usepackage[usenames,dvipsnames,table,xcdraw]{xcolor}

\newcommand{\ignore}[1]{}

\newcommand{\YES}{\textsc{Yes}\xspace}
\newcommand{\NO}{\textsc{No}\xspace}

\renewcommand{\AA}{{\mathcal A}}

\newcommand{\CC}{{\mathcal C}}
\newcommand{\DD}{{\mathcal D}}

\newcommand{\LL}{{\mathcal L}}

\newcommand{\PP}{{\mathcal P}}
\newcommand{\QQ}{{\mathcal Q}}

\renewcommand{\SS}{{\mathcal S}}
\newcommand{\TT}{{\mathcal T}}
\newcommand{\UU}{{\mathcal U}}
\newcommand{\VV}{{\mathcal V}}
\newcommand{\WW}{{\mathcal W}}
\newcommand{\XX}{{\mathcal X}}

\newcommand{\ppp}{{\mathfrak p}}

\newcommand{\vvv}{{\mathfrak v}}

\newcommand{\el}{\ensuremath{\ell}\xspace}
\newcommand{\prm}{\prime}

\newcommand{\NP}{\ensuremath{\mathsf{NP}}\xspace}

\newcommand{\coNPH}{\ensuremath{\mathsf{co}}-\ensuremath{\mathsf{NP}}-hard\xspace}
\newcommand{\NPH}{\ensuremath{\mathsf{NP}}-hard\xspace}
\newcommand{\NPC}{\ensuremath{\mathsf{NP}}-complete\xspace}

\newcommand{\Pb}{\ensuremath{\mathsf{P}}\xspace}

\newcommand{\WM}{\textsc{Weak Manipulation}\xspace}
\newcommand{\SM}{\textsc{Strong Manipulation}\xspace}
\newcommand{\OM}{\textsc{Opportunistic Manipulation}\xspace}
\newcommand{\CM}{\textsc{Coalitional Manipulation}\xspace}
\newcommand{\PW}{\textsc{Possible Winner}\xspace}
\newcommand{\NW}{\textsc{Necessary Winner}\xspace}

\usepackage{tikz}
\usetikzlibrary{shapes.geometric, arrows,decorations.text}

\usepackage{xspace}

\usepackage{multirow}

\usepackage[poorman]{cleveref}

\newtheorem{proposition}{\bf Proposition}
\newtheorem{observation}{\bf Observation}
\newtheorem{theorem}{\bf Theorem}
\newtheorem{lemma}{\bf Lemma}

\newtheorem{definition}{\bf Definition}

\crefname{observation}{Observation}{Observation}
\crefname{theorem}{Theorem}{Theorem}
\crefname{lemma}{Lemma}{Lemma}
\crefname{corollary}{Corollary}{Corollary}
\crefname{proposition}{Proposition}{Proposition}
\crefname{definition}{Definition}{Definition}
\crefname{claim}{Claim}{Claim}
\crefname{reductionrule}{Reduction rule}{Reduction rule}
\crefname{table}{Table}{Table}

\title{\bf Complexity of Manipulation with Partial Information in Voting}
\author{{Palash Dey}\\ 
Tata Institute of Fundamental Research, Mumbai\\
palash.dey@tifr.res.in
\and
{Neeldhara Misra}\\
Indian Institute of Technology, Gandhinagar\\
mail@neeldhara.com
\and
{Y. Narahari}\\
Indian Institute of Science, Bangalore\\
hari@csa.iisc.ernet.in
}

\begin{document}

\maketitle

\begin{abstract}
The Coalitional Manipulation problem has been
studied extensively in the literature for many voting rules. 
However, most studies have focused on the complete information setting, wherein
the manipulators know the votes of the non-manipulators. While this assumption
is reasonable for purposes of showing intractability, it is unrealistic for algorithmic
considerations. In most real-world scenarios, it is impractical to assume that the manipulators to have accurate knowledge of all the other votes. In this work, we investigate
manipulation with incomplete information. In our framework,
the manipulators know a partial order for each voter that is consistent with
the true preference of that voter. In this setting, we formulate three natural
computational notions of manipulation, namely weak, opportunistic, and strong manipulation. 
We say that an extension of a partial order is
\textit{viable} if there exists a manipulative vote for that extension. We propose the following notions of manipulation when manipulators have incomplete information about the votes of other voters.

\begin{enumerate}
\item \textsc{Weak Manipulation}: the manipulators seek to
vote in a way that makes their preferred candidate win in \textit{at least one
extension} of the partial votes of the non-manipulators.
\item \textsc{Opportunistic Manipulation}: the manipulators seek to
vote in a way that makes their preferred candidate win\textit{ in every
viable extension} of the partial votes of the non-manipulators.
\item \textsc{Strong Manipulation}: the manipulators seek to
vote in a way that makes their preferred candidate win \textit{in every
extension} of the partial votes of the non-manipulators.
\end{enumerate}

We consider several scenarios for which the traditional manipulation problems are easy (for instance, Borda with a single manipulator). For many of them, the corresponding manipulative questions that we propose turn out to be computationally intractable. Our hardness results often hold even when very little information is missing, or in other words, even when the instances are very close to the complete information setting. Our results show that the impact of paucity of information on the computational complexity of manipulation crucially depends on the notion of manipulation under consideration. Our overall conclusion is that computational hardness continues to be a valid obstruction to manipulation, in the context of a more realistic model.
\end{abstract}


\section{Introduction}

In many real life and AI related applications, agents often need to agree upon a common decision although they have different preferences over the available alternatives. A natural tool used in these situations is voting. Some classic examples of the use of voting rules in the context of multiagent systems include Clarke tax~\cite{ephrati1991clarke}, collaborative filtering~\cite{PennockHG00}, and similarity search~\cite{Fagin}, etc.
In a typical voting scenario, we have a set of candidates and a set of voters 
reporting their rankings of the candidates called their preferences or votes. A voting rule 
selects one candidate as the winner once all voters provide their votes. A set of votes over a set of candidates along with a voting rule is called an election.
A central issue in voting is the possibility of \emph{manipulation}. For many voting rules, it turns out that even a single vote, if cast differently, can alter the outcome. In particular, a voter manipulates an election if, by misrepresenting her preference, she obtains an outcome that she prefers over the ``honest'' outcome. In a cornerstone impossibility result, Gibbard and Satterthwaite~\cite{gibbard1973manipulation,satterthwaite1975strategy} show that every unanimous and non-dictatorial voting rule with three candidates or more is manipulable. We refer to~\cite{brandt2015handbook} for an excellent introduction to various strategic issues in computational social choice theory.

Considering that voting rules are indeed susceptible to manipulation, it is natural to seek ways by which elections can be protected from manipulations. The works of Bartholdi et al.~\cite{bartholdi1989computational,bartholdi1991single} approach the problem from the perspective of computational intractability. They exploit the possibility that voting rules, despite being vulnerable to manipulation in theory, may be hard to manipulate in practice. Indeed, a manipulator is faced with the following decision problem: given a collection of votes $\mathcal{P}$ and a distinguished candidate $c$, does there exist a vote $v$ that, when tallied with $\mathcal{P}$, makes $c$ win for a (fixed) voting rule $r$? The manipulation problem has subsequently been generalized to the problem of \textsc{Coalitional manipulation} by Conitzer et al.~\cite{conitzer2007elections}, where one or more manipulators collude together and try to make a distinguished candidate win the election. The manipulation problem, fortunately, turns out to be \NP{}-hard in several settings. This established the success of the approach of demonstrating a computational barrier to manipulation.

However, despite having set out to demonstrate the hardness of manipulation, the initial results in~\cite{bartholdi1989computational} were to the contrary, indicating that many voting rules are in fact easy to manipulate. Moreover, even with multiple manipulators involved, popular voting rules like plurality, veto, $k$-approval, Bucklin, and Fallback continue to be easy to manipulate~\cite{xia2009complexity}. While we know that the computational intractability may not provide a strong barrier \cite{procaccia2006junta,ProcacciaR07,friedgut2008elections,xia2008generalized,xia2008sufficient,faliszewski2010using,walsh2010empirical,walsh2011hard,isaksson2012geometry,dey2015computational,DeyMN15,journalsDeyMN16,DeyMN15a,dey2014asymptotic,dey2015asymptoticjournal} even for rules for which the coalitional manipulation problem turns out to be \NP{}-hard, in all other cases the possibility of manipulation is a much more serious concern. 

\subsection{Motivation and Problem Formulation}

In our work, we propose to extend the argument of computational intractability to address the cases where the approach appears to fail. We note that most incarnations of the manipulation problem studied so far are in the complete information setting, where the manipulators have complete knowledge of the preferences of the truthful voters. While these assumptions are indeed the best possible for the computationally negative results, we note that they are not reflective of typical real-world scenarios. Indeed, concerns regarding privacy of information, and in other cases, the sheer volume of information, would be significant hurdles for manipulators to obtain complete information. Motivated by this, we consider the manipulation problem in a natural \emph{partial information} setting. In particular, we model the partial information of the manipulators about the votes of the non-manipulators as partial orders over the set of candidates. A partial order over the set of candidates will be called a partial vote. Our results show that several of the voting rules that are easy to manipulate in the complete information setting become intractable when the manipulators know only partial votes. Indeed, for many voting rules, we show that even if the ordering of a small number of pairs of candidates is missing from the profile, manipulation becomes an~intractable problem. Our results therefore strengthen the view that manipulation may not be practical if we limit the information the manipulators have at their disposal about the votes of other voters \cite{conitzer2011dominating}.

We introduce three new computational problems that, in a natural way, extend the question of manipulation to the partial information setting. In these problems, the input is a set of partial votes $\mathcal{P}$ corresponding to the votes of the non-manipulators, a non-empty set of manipulators $M$, and a preferred candidate $c$. The task in the \textsc{Weak Manipulation (WM)} problem is to determine if there is a way to cast the manipulators' votes such that $c$ wins the election for at least one extension of the partial votes in $\mathcal{P}$. On the other hand, in the \textsc{Strong Manipulation (SM)} problem, we would like to know if there is a way of casting the manipulators' votes such that $c$ wins the election in \emph{every extension} of the partial votes in $\mathcal{P}$. 

We also introduce the problem of \textsc{Opportunistic Manipulation (OM)}, which is an ``intermediate'' notion of manipulation. Let us call an extension of a partial profile \textit{viable} if it is possible for the manipulators to vote in such a way that the manipulators' desired candidate wins in that extension. In other words, a viable extension is a \YES-instance of the standard \CM problem. We have an opportunistic manipulation when it is possible for the manipulators to cast a vote which makes $c$ win the election in \emph{ all} viable extensions. Note that any \YES-instance of \SM is also an \YES-instance of \OM, but this may not be true in the reverse direction. As a particularly extreme example, consider a partial profile where there are no viable extensions: this would be a \NO-instance for \SM, but a (vacuous) \YES-instance of \OM. The \OM problem allows us to explore a more relaxed notion of manipulation: one where the manipulators are obliged to be successful only in extensions where it is possible to be successful. Note that the goal with \SM is to be successful in all extensions, and therefore the only interesting instances are the ones where all extensions are viable. 

It is easy to see that \YES{} instance of \SM is also a \YES{} instance of \OM and \WM. Beyond this, we remark that all the three problems are questions with different goals, and neither of them render the other redundant. We refer the reader to~\Cref{Fig:EG} for a simple example distinguishing these scenarios.

All the problems above generalize \CM, and hence any computational intractability result for \CM immediately yields a corresponding intractability result for \WM, \SM, and \OM under the same setting. For example, it is known that the \CM problem is intractable for the maximin voting rule when we have at least two manipulators~\cite{xia2009complexity}. Hence, the \WM, \SM, and \OM problems are intractable for the maximin voting rule when we have at least two manipulators.

\begin{figure}[t]
\centering
\includegraphics[scale=0.3]{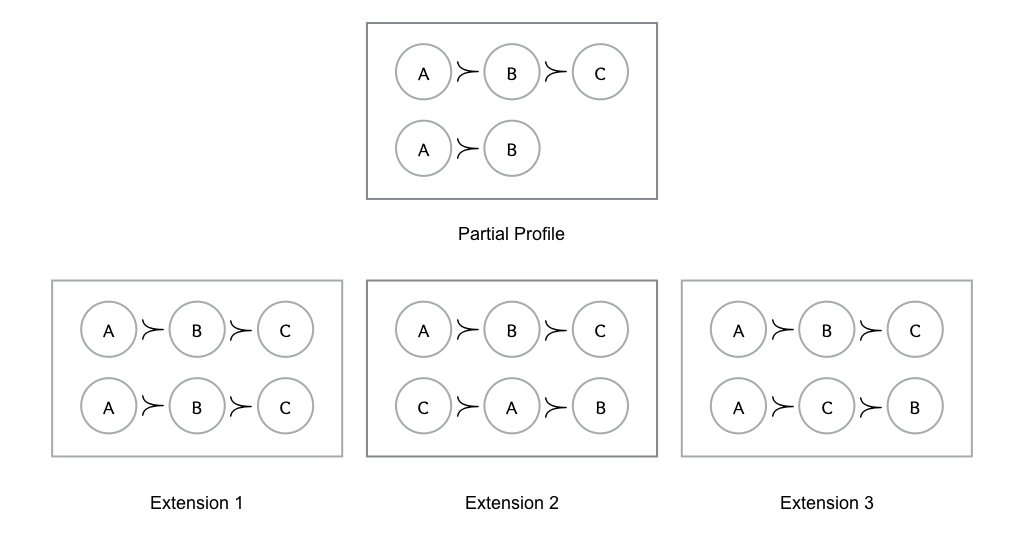}
\caption{An example of a partial profile. Consider the plurality voting rule with one manipulator. If the favorite candidate is A, then the manipulator simply has to place A on the top of his vote to make A win in \textit{any} extension. If the favorite candidate is B, there is \textit{no vote} that makes B win in any extension. Finally, if the favorite candidate is C, then with a vote that places C on top, the manipulator can make C win in the only viable extension (Extension 2).}
\label{Fig:EG}	
\end{figure}

\subsection{Related Work}

A notion of manipulation under partial information has been considered by Conitzer et al.~\cite{conitzer2011dominating}. They focus on whether or not there exists a dominating manipulation and show that this problem is \NP{}-hard for many common voting rules. Given some partial votes, a dominating manipulation is a non-truthful vote that the manipulator can cast which makes the winner at least as preferable (and sometimes more preferable) as the winner when the manipulator votes truthfully. The dominating manipulation problem and the \WM, \OM, and \SM problems do not seem to have any apparent complexity-theoretic connection. For example, the dominating manipulation problem is \NP{}-hard for all the common voting rules except plurality and veto, whereas, the \SM problem is easy for most of the cases (see \Cref{table:partial_summary}). However, the results in \cite{conitzer2011dominating} establish the fact that it is indeed possible to make manipulation intractable by restricting the amount of information the manipulators possess about the votes of the other voters. Elkind and Erd{\'e}lyi~\cite{elkind2012manipulation} study manipulation under voting rule uncertainty. However, in our work, the voting rule is fixed and known to the manipulators.

Two closely related problems that have been extensively studied in the context of incomplete votes are \textsc{Possible Winner} and \textsc{Necessary Winner}~\cite{konczak2005voting}. In the \PW problem, we are given a set of partial votes $\mathcal{P}$ and a candidate $c$, and the question is whether there exists an extension of $\mathcal{P}$ where $c$ wins, while in the \NW problem, the question is whether $c$ is a winner in every extension of $\mathcal{P}$. Following the work in~\cite{konczak2005voting}, a number of  special cases and variants of the \PW problem have been studied in the literature~\cite{chevaleyre2010possible,bachrach2010probabilistic,baumeister2011computational,baumeister2012possible,gaspers2014possible,xia2008determining,ding2013voting,NarodytskaW14,BaumeisterFLR12,MenonL15}. The flavor of the \WM problem is clearly similar to \PW. However, we emphasize that there are subtle distinctions between the two problems. A more elaborate comparison is made in the next section. 

\subsection{Our Contribution}

Our primary contribution in this work is to propose and study three natural and realistic generalizations of the computational problem of manipulation in the incomplete information setting. We summarize the complexity results in this work in \Cref{table:partial_summary}. Our results provide the following interesting insights on the impact of lack of information on the computational difficulty of manipulation. We note that the number of undetermined pairs of candidates per vote are small constants in all our hardness results.

\begin{itemize}
 \item We observe that the computational problem of manipulation for the plurality and veto voting rules remains polynomial time solvable even with lack of information, irrespective of the notion of manipulation under consideration [\Cref{pv_easy,thm:vetoOM,sm_k_easy,obs:plurality_fallback}]. We note that the plurality and veto voting rule also remain vulnerable under the notion of dominating manipulation~\cite{conitzer2011dominating}.
 
 \item The impact of absence of information on the computational complexity of manipulation is more dynamic for the $k$-approval, $k$-veto, Bucklin, Borda, and maximin voting rules. Only the \WM and \OM problems are computationally intractable for the $k$-approval [\Cref{wm_k_hard,thm:kappOM}], $k$-veto [\Cref{k_veto_wm,thm:kvetoOM}], Bucklin [\Cref{wm_bucklin_hard,thm:bucklinOM}], Borda [\Cref{cm_hard_one,thm:bordaOM}], and maximin [\Cref{cm_hard_one,thm:maximinOM}] voting rules, whereas the \SM problem remains computationally tractable [\Cref{sm_k_easy,sm_sr_easy,sm_bucklin_easy,sm_maximin_easy}].
 
 \item \Cref{table:partial_summary} shows an interesting behavior of the Fallback voting rule. The Fallback voting rule is the only voting rule among the voting rules we study here for which the \WM problem is NP-hard [\Cref{wm_bucklin_hard}] but both the \OM and \SM problems are polynomial time solvable [\Cref{sm_bucklin_easy,obs:plurality_fallback}]. This is because the \OM problem can be solved for the Fallback voting rule by simply making manipulators vote for their desired candidate.
 
 \item Our results show that absence of information makes all the three notions of manipulations intractable for the Copeland$^\alpha$ voting rule for every rational $\alpha\in[0,1]\setminus\{0.5\}$ for the \WM problem [\Cref{cm_hard_one}] and for every $\alpha\in[0,1]$ for the \OM and \SM problems [\Cref{sm_copeland_hard,thm:copelandOM}].
\end{itemize}

Our results (see \Cref{table:partial_summary}) show that whether lack of information makes the manipulation problems harder, crucially depends on the notion of manipulation applicable to the situation under consideration. All the three notions of manipulations are, in our view, natural extension of manipulation to the incomplete information setting and tries to capture different behaviors of manipulators. For example, the \WM problem maybe applicable to an optimistic manipulator whereas for an pessimistic manipulator, the \SM problem may make more sense.


\begin{table}
\centering
\resizebox{\linewidth}{!}{
\renewcommand*{\arraystretch}{1.2}
\begin{tabular}{|c|c|c|c|c|c|c|}
\hline
                  & WM, $\ell = 1$                                                          & WM                                                         & OM, $\ell = 1$                            & OM                           & SM, $\ell = 1$                              & SM                                                                                                          \\ \hline
Plurality         & \multicolumn{2}{c}{}                                                                                        & \multicolumn{2}{c}{\cellcolor[HTML]{67FD9A}}                            & \multicolumn{2}{c|}{\cellcolor[HTML]{67FD9A}}                                                                                                             \\ \cline{1-1}
Veto              & \multicolumn{2}{c}{\multirow{-2}{*}{P }}                                                     & \multicolumn{2}{c}{\multirow{-2}{*}{\cellcolor[HTML]{67FD9A}P}}         & \multicolumn{2}{c|}{\cellcolor[HTML]{67FD9A}}                                                                                                             \\ \cline{1-1}

$k$-Approval      & \multicolumn{2}{c}{\cellcolor[HTML]{FE996B}}                                                                                        & \multicolumn{2}{c}{\cellcolor[HTML]{FFCB2F}}                            & \multicolumn{2}{c|}{\cellcolor[HTML]{67FD9A}}                                                                                                             \\ \cline{1-1}
$k$-Veto          & \multicolumn{2}{c}{\cellcolor[HTML]{FE996B}}                                                                                        & \multicolumn{2}{c}{\cellcolor[HTML]{FFCB2F}}                            & \multicolumn{2}{c|}{\cellcolor[HTML]{67FD9A}}                                                                                                             \\ \cline{1-1}
Bucklin           & \multicolumn{2}{c}{\cellcolor[HTML]{FE996B}}                                                                                        & \multicolumn{2}{c}{\multirow{-3}{*}{\cellcolor[HTML]{FFCB2F}coNP-hard}} & \multicolumn{2}{c|}{\cellcolor[HTML]{67FD9A}}                                                                                                             \\ \cline{1-1}
Fallback          & \multicolumn{2}{c}{\multirow{-4}{*}{\cellcolor[HTML]{FE996B}NP-complete}}                                                           & \multicolumn{2}{c}{\cellcolor[HTML]{67FD9A}P}                           & \multicolumn{2}{c|}{\multirow{-6}{*}{\cellcolor[HTML]{67FD9A}P}}                                                                                          \\ \cline{1-1}
Borda             & \multicolumn{2}{c}{}                                                                                        & \multicolumn{2}{c}{\cellcolor[HTML]{FFCB2F}}                            & \multicolumn{1}{c}{\cellcolor[HTML]{67FD9A}}                    &                                                                                 \\ \cline{1-1}
maximin           & \multicolumn{2}{c}{}                                                                                        & \multicolumn{2}{c}{\cellcolor[HTML]{FFCB2F}}                            & \multicolumn{1}{c}{\multirow{-2}{*}{\cellcolor[HTML]{67FD9A}P}} &                                                                                     \\ \cline{1-1}
Copeland$^\alpha$ & \multicolumn{2}{c}{\multirow{-3}{*}{
\begin{tabular}[c]{@{}c@{}}NP-complete\\ \end{tabular}}} & \multicolumn{2}{c}{\multirow{-3}{*}{\cellcolor[HTML]{FFCB2F}coNP-hard}} & \multicolumn{1}{c}{\cellcolor[HTML]{FFCB2F}coNP-hard}         & \multirow{-3}{*}{\begin{tabular}[c]{@{}c@{}}NP-hard\\ \end{tabular}} \\ \hline
\end{tabular}}
\caption{Summary of Results ($\el$ denotes the number of manipulators). The results in white follow immediately from the literature (\Cref{pw_hard,cm_hard,cm_hard_one}). Our results for the Copeland$^\alpha$ voting rule hold for every rational $\alpha\in[0,1]\setminus\{0.5\}$ for the \WM problem and for every $\alpha\in[0,1]$ for the \OM and \SM problems.}
\label{table:partial_summary}
\end{table}

\paragraph{Organization of the paper:} We define the problems and introduce the basic terminology in the next section. We present our hardness results in Section~\ref{sec:hard}. In~\Cref{sec:poly}, we present our polynomially solvable algorithms. Finally we conclude with future directions of research in Section~\ref{sec:con}.

\section{Preliminaries}\label{sec:def}

In this section, we begin by providing the technical definitions and notations that we will need in the subsequent sections. We then formulate the problems that capture our notions of manipulation when the votes are given as partial orders, and finally draw comparisons with related problems that are already studied in the literature of computational social choice theory. 

\subsection{Notations and Definitions}

Let $\mathcal{V}=\{v_1, \dots, v_n\}$ be the set of all \emph{voters} and $\mathcal{C}=\{c_1, \dots, c_m\}$ the set of all \emph{candidates}. If not specified explicitly, $n$ and $m$ denote the total number of voters and the total number of candidates respectively. Each voter $v_i$'s \textit{vote} is a \emph{preference} $\succ_i$ over the candidates which is a linear order over $\mathcal{C}$. For example, for two candidates $a$ and $b$, $a \succ_i b$ means that the voter $v_i$ prefers $a$ to $b$. We denote the set of all linear orders over $\mathcal{C}$ by $\mathcal{L(C)}$. Hence, $\mathcal{L(C)}^n$ denotes the set of all $n$-voters' preference profile $(\succ_1, \dots, \succ_n)$. 
A map $r:\cup_{n,|\mathcal{C}|\in\mathbb{N}^+}\mathcal{L(C)}^n\longrightarrow 2^\mathcal{C}\setminus\{\emptyset\}$
is called a \emph{voting rule}. For some preference profile $\succ\,\in \mathcal{L(C)}^n$, if $r(\succ)=\{w\}$, then we say $w$ wins uniquely and we write $r(\succ)=w$. From here on, whenever we say some candidate $w$ wins, we mean that the candidate $w$ wins uniquely. For simplicity, we restrict ourselves to the unique winner case in this paper. All our proofs can be easily extended in the co-winner case. 

A more general setting is an {\em election\/} where the votes are only \emph{partial orders} over candidates. A \emph{partial order} is a relation that is \emph{reflexive, antisymmetric}, and \emph{transitive}. A partial vote can be extended to possibly more than one linear vote depending on how we fix the order of the unspecified pairs of candidates. For example, in an election with the set of candidates $\mathcal{C} = \{a, b, c\}$, a valid partial vote can be $a \succ b$. This partial vote can be extended to three linear votes namely, $a \succ b \succ c$, $a \succ c \succ b$, $c \succ a \succ b$. In this paper, we often define a partial vote like $\succ \setminus A$, where $\succ\,\in\mathcal{L(C)}$ and $A\subset \mathcal{C}\times\mathcal{C}$, by which we mean the partial vote obtained by removing the order among the pair of candidates in $A$ from $\succ$. Also, whenever we do not specify the order among a set of candidates while describing a complete vote, the statement/proof is correct in whichever way we fix the order among them. We now give examples of some common voting rules.

\begin{itemize}
 \item {\bf Positional scoring rules:} An $m$-dimensional vector $\vec{\alpha}=\left(\alpha_1,\alpha_2,\dots,\alpha_m\right)\in\mathbb{R}^m$ 
 with $\alpha_1\ge\alpha_2\ge\dots\ge\alpha_m$ and $\alpha_1>\alpha_m$ naturally defines a
 voting rule -- a candidate gets score $\alpha_i$ from a vote if it is placed at the $i^{th}$ position, and the 
 score of a candidate is the sum of the scores it receives from all the votes. 
 The winners are the candidates with maximum score. Scoring rules remain unchanged if we multiply every $\alpha_i$ by any constant $\lambda>0$ and/or add any constant $\mu$. Hence, we assume without loss of generality that for any score vector $\vec{\alpha}$, there exists a $j$ such that $\alpha_j - \alpha_{j+1}=1$ and $\alpha_k = 0$ for all $k>j$. We call such a $\vec{\alpha}$ a normalized score vector.
 For $\vec{\alpha}=\left(m-1,m-2,\dots,1,0\right)$, we get the \emph{Borda} voting rule. With $\alpha_i=1$ $\forall i\le k$ and $0$ else, 
 the voting rule we get is known as $k$-\emph{approval}. For the $k$-\emph{veto} voting rule, we have  $\alpha_i=0$ $\forall i\le m-k$ and $-1$ else. \emph{Plurality} is $1$-\emph{approval} and \emph{veto} is $1$-\emph{veto}.
 
 \item {\bf Bucklin and simplified Bucklin:} Let \el be the minimum integer such that at least one candidate gets majority within top \el positions of the votes. The winners under the simplified Bucklin voting rule are the candidates having more than $\nfrac{n}{2}$ votes within top \el positions. The winners under the Bucklin voting rule are the candidates appearing within top \el positions of the votes highest number of times.
 
 \item {\bf Fallback and simplified Fallback:} For these voting rules, each voter $v$ ranks a subset $\XX_v\subset\CC$ of candidates and disapproves the rest of the candidates~\cite{brams2009voting}. Now for the Fallback and simplified Fallback voting rules, we apply the Bucklin and simplified Bucklin voting rules respectively to define winners. If there is no integer \el for which at least one candidate gets more than $\nfrac{n}{2}$ votes, both the Fallback and simplified Fallback voting rules output the candidates with most approvals as winners. We assume, for simplicity, that the number of candidates each partial vote approves is known.
 
 \item {\bf Maximin:} For any two candidates $x$ and $y$, let $D(x,y)$ be $N(x,y) - N(y,x)$, where $N(x,y)$ $(\text{respectively }N(y,x))$ is the number of voters who prefer $x$ to $y$ (respectively $y$ to $x$). The election we get by restricting all the votes to $x$ and $y$ only is called the pairwise election between $x$ and $y$. The maximin score of a candidate $x$ is $\min_{y\ne x} D(x,y)$. The winners are the candidates with maximum maximin score.
 
 \item {\bf Copeland$^\alpha$.} The Copeland$^\alpha$ score of a candidate $x$ is $|\{y\ne x:D_{\mathcal{E}}(x,y)>0\}|+\alpha|\{y\ne x:D_{\mathcal{E}}(x,y)=0\}|$, where $\alpha\in [0,1]$. That is, the Copeland$^\alpha$ of a candidate $x$ is the number of other candidates it defeats in pairwise election plus $\alpha$ times the number of other candidates it ties with in pairwise elections. The winners are the candidates with the maximum Copeland$^\alpha$ score. 
\end{itemize}

\subsection{Problem Definitions}\label{sec:probdef}

We now formally define the three problems that we consider in this work, namely \textsc{Weak Manipulation},  \textsc{Opportunistic Manipulation}, and \textsc{Strong Manipulation}. Let $r$ be a fixed voting rule. We first introduce the \textsc{Weak Manipulation} problem.

\begin{definition}\textbf{\textsc{$r$-Weak Manipulation}}\\
 Given a set of partial votes $\mathcal{P}$ over a set of candidates $\mathcal{C}$, a positive integer $\el~(>0)$ denoting the number of manipulators, and a candidate $c$, do there exist votes $\succ_1, \ldots, \succ_\el\, \in \mathcal{L(\mathcal{C})}$ such that there exists an extension $\succ\, \in \mathcal{\mathcal{L(\mathcal{C})}^{|\mathcal{P}|}}$ of $\mathcal{P}$ with $r(\succ, \succ_1, \ldots, \succ_\el) = c$?
\end{definition}

To define the \textsc{Opportunistic Manipulation} problem, we first introduce the notion of an $(r,c)$-opportunistic voting profile, where $r$ is a voting rule and $c$ is any particular candidate.

\begin{definition}\textbf{$(r,c)$-Opportunistic Voting Profile}\\
 Let $\el$ be the number of manipulators and $\PP$ a set of partial votes. An $\el$-voter profile $(\succ_i)_{i\in[\el]}\in\LL(\CC)^\el$ is called an $(r,c)$-opportunistic voting profile if for each extension $\overline{\PP}$ of $\PP$ for which there exists an $\el$-vote profile $(\succ^\prime_i)_{i\in[\el]}\in\LL(\CC)^\el$ with $r\left(\overline{\PP}\cup\left(\succ^\prime_i\right)_{i\in[\el]}\right) = c$, we have $r\left(\overline{\PP}\cup\left(\succ_i\right)_{i\in[\el]}\right) = c$.
\end{definition}

In other words, an $\ell$-vote profile is $(r,c)$-opportunistic with respect to a partial profile if, when put together with the truthful votes of any extension,  $c$ wins if the extension is viable to begin with. We are now ready to define the \textsc{Opportunistic Manipulation} problem.

\begin{definition}\textbf{\textsc{$r$-Opportunistic Manipulation}}\\
 Given a set of partial votes $\mathcal{P}$ over a set of candidates $\mathcal{C}$, a positive integer $\el~(>0)$ denoting the number of manipulators, and a candidate $c$, does there exist an $(r,c)$-opportunistic $\el$-vote profile?
\end{definition}

We finally define the \textsc{Strong Manipulation} problem.

\begin{definition}\textbf{\textsc{$r$-Strong Manipulation}}\\
 Given a set of partial votes $\mathcal{P}$ over a set of candidates $\mathcal{C}$, a positive integer $\el~(>0)$ denoting the number of manipulators, and a candidate $c$, do there exist votes $(\succ_i)_{i\in\el} \in \mathcal{L(\mathcal{C})}^\el$ such that for every extension $\succ\, \in \mathcal{\mathcal{L(\mathcal{C})}^{|\mathcal{P}|}}$ of $\mathcal{P}$, we have $r(\succ, (\succ_i)_{i\in[\el]}) = c$?
\end{definition}

We use $(\PP, \el, c)$ to denote instances of \WM, \OM, and \SM, where $\PP$ denotes a profile of partial votes, $\el$ denotes the number of manipulators, and $c$ denotes the desired winner.

For the sake of completeness, we provide the definitions of the \textsc{Coalitional Manipulation} and \textsc{Possible Winner} problems below.
\begin{definition}\textbf{\textsc{$r$-Coalitional Manipulation}}\\
 Given a set of complete votes $\succ$ over a set of candidates $\mathcal{C}$, a positive integer $\el~(>0)$ denoting the number of manipulators, and a candidate $c$, do there exist votes $(\succ_i)_{i\in\el} \in \mathcal{L(\mathcal{C})}^\el$ such that $r\left(\succ, \left(\succ_i\right)_{i\in[\el]}\right) = c$?
\end{definition}

\begin{definition}\textbf{\textsc{$r$-Possible Winner}}\\
 Given a set of partial votes $\mathcal{P}$ and a candidate $c$, does there exist an extension $\succ$ of the partial votes in $\mathcal{P}$ to linear votes such that $r(\succ)=c$?
\end{definition}

\subsection{Comparison with Possible Winner and Coalitional Manipulation} ~For any fixed voting rule, the \WM problem with $\ell$ manipulators reduces to the \PW problem. This is achieved by simply using the same set as truthful votes and introducing $\ell$ empty votes. We summarize this in the observation below.

\begin{observation}\label{pw_hard}
 The \WM problem many-to-one reduces to the \PW problem for every voting rule.
\end{observation}

\begin{proof}
 Let $(\PP, \el, c)$ be an instance of \WM. Let $\QQ$ be the set consisting of \el many copies of partial votes $\{\emptyset\}$. Clearly the \WM instance $(\PP, \el, c)$ is equivalent to the \PW instance $(\PP\cup\QQ, c)$.
\end{proof}

However, whether the \PW problem reduces to the \WM problem or not is not clear since in any \WM problem instance, there must exist at least one manipulator and a \PW instance may have no empty vote. From a technical point of view, the difference between the \WM and \PW problems may look marginal; however we believe that the \WM problem is a very natural generalization of the \CM problem in the partial information setting and thus worth studying. Similarly, it is easy to show, that the \CM problem with $\ell$ manipulators reduces to \WM, \OM, and \SM problems with $\ell$ manipulators, since the former is a special case of the latter ones.

\begin{observation}\label{cm_hard}
 The \CM problem with \el manipulators many-to-one reduces to \WM, \OM, and \SM problems with \el manipulators for all voting rules and for all positive integers \el.
\end{observation}

\begin{proof}
 Follows from the fact that every instance of the \CM problem is also an equivalent instance of the \WM, \OM, and \SM problems.
\end{proof}

Finally, we note that the \CM problem with $\ell$ manipulators can be reduced to the \WM problem with just one manipulator, by introducing $\ell-1$ empty votes. These votes can be used to witness a good extension in the forward direction. In the reverse direction, given an extension where the manipulator is successful, the extension can be used as the manipulator's votes. This argument leads to the following observation.

\begin{observation}\label{cm_hard_one}
 The \CM problem with \el manipulators many-to-one reduces to the \WM problem with one manipulator for every voting rule and for every positive integer \el.
\end{observation}

\begin{proof}
 Let $(\PP, \el, c)$ be an instance of \CM. Let $\QQ$ be the set of consisting of $\el-1$ many copies of partial vote $\{c\succ \text{others}\}$. Clearly the \WM instance $(\PP\cup\QQ, 1, c)$ is equivalent to the \CM instance $(\PP, \el, 1)$.
\end{proof}

This observation can be used to derive the hardness of \WM even for one manipulator whenever the hardness for \CM is known for any fixed number of manipulators (for instance, this is the case for the voting rules such as Borda, maximin and Copeland). However, determining the complexity of \WM with one manipulator requires further work for voting rules where \CM is polynomially solvable for any number of manipulators (such as $k$-approval, Plurality, Bucklin, and so on).

\section{Hardness Results}\label{sec:hard}

In this section, we present our hardness results. While some of our reductions are from the \textsc{Possible Winner} problem, the other reductions in this section are from the \textsc{Exact Cover by 3-Sets} problem, also referred to as X3C. This is a well-known \NPC{}~\cite{garey1979computers} problem, and is defined as follows.

\begin{definition}[Exact Cover by 3-Sets (X3C)] Given a set $\UU$ and a collection $\SS = \{S_1,S_2, \dots, S_t\}$ of $t$ subsets of $ \UU$ with $|S_i|=3 ~\forall i=1, \dots, t,$ does there exist a $\TT\subset\SS$ with $|\TT|=\frac{|\UU|}{3}$ such that $\cup_{X\in \TT} X = \UU$?
\end{definition}

We use $\overline{\text{X3C}}$ to refer to the complement of X3C, which is to say that an instance of $\overline{\text{X3C}}$ is a \YES instance if and only if it is a \NO instance of X3C. The rest of this section is organized according to the problems being addressed.

\subsection{Weak Manipulation} 

To begin with, recall that the \CM problem is \NPC for the Borda~\cite{davies2011complexity,betzler2011unweighted}, maximin~\cite{xia2009complexity}, and Copeland$^\alpha$~\cite{faliszewski2008copeland,FaliszewskiHHR09,faliszewski2010manipulation} voting rules for every rational $\alpha\in[0,1]\setminus\{0.5\}$, when we have two manipulators. Therefore, it follows from~\Cref{cm_hard_one} that the \WM problem is \NPC for the Borda, maximin, and Copeland$^\alpha$ voting rules for every rational $\alpha\in[0,1]\setminus\{0.5\}$, even with one manipulator. 

For the $k$-approval and $k$-veto voting rules, we reduce from the corresponding \PW problems. While it is natural to start from the same voting profile, the main challenge is in undoing the advantage that the favorite candidate receives from the manipulator's vote, in the reverse direction.

We begin with proving that the \WM problem is \NP{}-complete for the $k$-approval voting rule even with one manipulator and at most $4$ undetermined pairs per vote.

\begin{theorem}\label{wm_k_hard}
The \WM problem is \NP{}-complete for the $k$-approval voting rule even with one manipulator for any constant $k>1$, even when the number of undetermined pairs in each vote is no more than $4$.
\end{theorem}

\begin{proof}
For simplicity of presentation, we prove the theorem for $2$-approval.
We reduce from the \PW problem for $2$-approval which is \NP{}-complete~\cite{xia2008determining}, even when the number of undetermined pairs in each vote is no more than $4$. 
Let $\mathcal{P}$ be the set of partial votes in a \PW instance, and 
let ${\mathcal C} = \{c_1, \ldots, c_m, c\}$ be the set of candidates, 
where the goal is to check if there is an extension of $\mathcal{P}$ that makes $c$ win. For developing the instance of \WM, we need to ``reverse'' any advantage that the candidate $c$ obtains from the vote of the manipulator. Notice that the most that the manipulator can do is to increase the score of $c$ by one. Therefore, in our construction, we \textit{``artificially''} increase the score of all the other candidates by one, so that despite of the manipulator's vote, $c$ will win the new election if and only if it was a possible winner in the \PW instance. To this end, we introduce $(m+1)$ many \textit{dummy} candidates $d_1, \ldots, d_{m+1}$ and the complete votes:
$$w_i = c_i \succ d_i \succ \text{others},\text{ for every } i\in \{1, \dots, m\}$$
Further, we extend the given partial votes of the \PW instance to force the dummy candidates 
to be preferred least over the rest - by defining, for every $v_i \in \mathcal{P}$, the corresponding partial vote $v_i^{\prime}$ as follows.
$$v_i^\prime = v_i \cup \{{\mathcal C} \succ \{d_1, \ldots, d_{m+1}\}\}.$$ 
This ensures that all the dummy candidates do not receive any score 
from the modified partial votes corresponding to the partial votes of the \PW instance. Notice that since the number of undetermined pairs in $v_i$ is no more than $4$, the number of undetermined pairs in $v_i^\prime$ is also no more than $4$. Let $({\mathcal C^\prime},\mathcal{Q},c)$ denote this constructed \WM instance. We claim that the two instances are equivalent.

In the forward direction, suppose $c$ is a possible winner with respect to $\mathcal{P}$, and let $\overline{\mathcal{P}}$ be an extension where $c$ wins. Then it is easy to see that the manipulator can make $c$ win in some extension by placing $c$ and $d_{m+1}$ in the first two positions of her vote (note that the partial score of $d_{m+1}$ is zero in $\mathcal{Q}$). Indeed, consider the extension of $\mathcal{Q}$ obtained by mimicking the extension $\overline{\mathcal{P}}$ on the ``common'' partial votes, $\{v_i^\prime ~|~ v_i \in \mathcal{P}\}$. Notice that this is well-defined since $v_i$ and $v_i^\prime$ have exactly the same set of incomparable pairs. In this extension, the score of $c$ is strictly greater than the scores of all the other candidates, since the scores of all candidates in $\mathcal{C}$ is exactly one more than their scores in $\mathcal{P}$, and all the dummy candidates have a score of at most one. 

In the reverse direction, notice that the manipulator puts the candidates $c$ and $d_{m+1}$ in the top two positions without loss of generality. Now suppose the manipulator's vote $c\succ d_{m+1}\succ \text{others}$ makes $c$ win the election for an extension $\overline{\mathcal{Q}}$ of $\mathcal{Q}$. Then consider the extension $\overline{\mathcal{P}}$ obtained by restricting $\overline{\mathcal{Q}}$ to $\mathcal{C}$. Notice that the score of each candidate in $\mathcal{C}$ in this extension is one less than their scores in $\mathcal{Q}$. Therefore, the candidate $c$ wins this election as well, concluding the proof. 

The above proof can be imitated for any other constant values of $k$ by reducing it from the \PW problem for $k$-approval and introducing $(m+1)(k-1)$ dummy candidates.
\end{proof}

We will use \Cref{score_gen} in subsequent proofs which has been used before~\cite{baumeister2011computational,DeyMN15,journalsDeyMN16}.

\begin{lemma}\label{score_gen}
Let $\mathcal{C} = \{c_1, \ldots, c_m\} \uplus D, (|D|>0)$ be a set of candidates, and $\vec{\alpha}$ a normalized score vector of length $|\mathcal{C}|$. Then, for any given $\mathbf{X} = (X_1, \ldots, X_m) \in \mathbb{Z}^m$, there exists $\lambda\in \mathbb{R}$ and a voting profile such that the $\vec{\alpha}$-score of $c_i$ is $\lambda + X_i$ for all $1\le i\le m$,  and the score of candidates $d\in D$ is less than $\lambda$. Moreover, the number of votes is $O(poly(|\mathcal{C}|\cdot \sum_{i=1}^m |X_i|, \lambda))$.
\end{lemma}
  
Note that the number of votes used in \Cref{score_gen} is polynomial in $m$ if $\lambda$ and $|X_i|$ is polynomial in $m$ for every $i\in[m]$, which indeed is the case in all the proofs that use \Cref{score_gen}. We next show that the WM problem is \NP{}-complete for the $k$-veto voting rule.

\begin{theorem}\label{k_veto_wm}
 The \WM problem for the $k$-veto voting rule is \NP{}-complete even with one manipulator for any constant $k>1$.
\end{theorem}

\begin{proof}
We reduce from the \PW problem for the $k$-veto voting rule which is known to be \NPC~\cite{betzler2009towards}. Let $\mathcal{P}$ be the set of partial votes in a \PW problem instance, and 
let ${\mathcal C} = \{c_1, \ldots, c_m, c\}$ be the set of candidates, 
where the goal is to check if there is an extension that makes $c$ win with respect to $k$-veto. We assume without loss of generality that $c$'s position is fixed in all the partial votes (if not, then we fix the position of $c$ as high as possible in every vote).

We introduce $k+1$ many \textit{dummy} candidates $d_1, \ldots, d_k, d$. The role of the first $k$ dummy candidates is to ensure that the manipulator is forced to place them at the ``bottom $k$'' positions of her vote, so that all the original candidates get the same score from the additional vote of the manipulator. The most natural way of achieving this is to ensure that the dummy candidates have the same score as $c$ in any extension (note that we know the score of $c$ since $c$'s position is fixed in all the partial votes). This would force the manipulator to place these $k$ candidates in the last $k$ positions. Indeed, doing anything else will cause these candidates to tie with $c$, even when there is an extension of $\mathcal{P}$ that makes $c$ win.

To this end, we begin by placing the dummy candidates in the top $k$ positions in all the partial votes. Formally, we modify every partial vote as follows:
$$w = d_i \succ \text{others},\text{ for every } i\in \{1, \dots, k\}$$
At this point, we know the scores of $c$ and $d_i,\text{ for every } i\in \{1, \dots, k\}$. Using \Cref{score_gen}, we add complete votes such that the final score of $c$ is the same with the score of every $d_i$ and the score of $c$ is strictly more than the score of $d$. The relative score of every other candidate remains the same. This completes the description of the construction. We denote the augmented set of partial votes by $\overline{\mathcal{P}}$. 

We now argue the correctness. In the forward direction, if there is an extension of the votes that makes $c$ win, then we repeat this extension, and the vote of the manipulator puts the candidate $d_i$ at the position $m+i+2$; and all the other candidates in an arbitrary fashion. Formally, we let the manipulator's vote be:
$$\mathfrak{v} = c \succ c_1 \succ \cdots \succ c_m \succ d \succ d_1 \succ \cdots \succ d_k.$$
By construction $c$ wins the election in this particular setup. In the reverse direction, consider a vote of the manipulator and an extension $\overline{\mathcal{Q}}$ of $\overline{\mathcal{P}}$ in which $c$ wins. Note that the manipulator's vote necessarily places the candidates $d_i$ in the bottom $k$ positions --- indeed, if not, then $c$ cannot win the election by construction. We extend a partial vote $w \in \mathcal{P}$ by mimicking the extension of the corresponding partial vote $w^\prime \in \overline{\mathcal{P}}$, that is, we simply project the extension of $w^\prime$ on the original set of candidates $\mathcal{C}$. Let $\mathcal{Q}$ denote this proposed extension of $\mathcal{P}$. We claim that $c$ wins the election given by $\mathcal{Q}$. Indeed, suppose not. Let $c_i$ be a candidate whose score is at least the score of $c$ in the extension $\mathcal{Q}$. Note that the scores of $c_i$ and $c$ in the extension  $\overline{\mathcal{Q}}$ are exactly the same as their scores in $\mathcal{Q}$, except for a constant offset --- importantly, their scores are offset by the same amount. This implies that the score of $c_i$ is at least the score of $c$ in $\overline{\mathcal{Q}}$ as well, which is a contradiction. Hence, the two instances are equivalent.
\end{proof}

We next prove, by a reduction from X3C, that the \textsc{Weak Manipulation} problem for the Bucklin and simplified Bucklin voting rules is \NPC{} even with one manipulator and at most $16$ undetermined pairs per vote.

\begin{theorem}\label{wm_bucklin_hard}
The \textsc{Weak Manipulation} problem is \NPC{} for Bucklin, simplified Bucklin, Fallback, and simplified Fallback voting rules, even when we have only one manipulator and the number of undetermined pairs in each vote is no more than $16$. 
\end{theorem}

\begin{proof}
We reduce the X3C problem to \textsc{Weak Manipulation} for simplified Bucklin. Let $(\UU = \{u_1, \ldots, u_m\}, \SS:= \{S_1,S_2, \dots, S_t\})$ be an instance of X3C, where each $S_i$ is a subset of $\UU$ of size three. We construct a \WM instance based on $(\UU,\SS)$ as follows.
$$ \text{Candidate set: } \mathcal{C} = \WW\cup \XX \cup \DD\cup \UU\cup \{c,w, a, b\},\text{ where } |\WW|=m-3, |\XX|=4, |\DD|=m+1$$ 
We first introduce the following partial votes $\PP$ in correspondence with the sets in the family as follows.
$$ \WW\succ \XX\succ S_i\succ c\succ (\UU\setminus S_i)\succ \DD \setminus \left(\{\XX \times (\{c\}\cup S_i)\}\right), \forall i\le t$$
Notice that the number of undetermined pairs in every vote in $\PP$ is $16$. We introduce the following additional complete votes $\QQ$:
\begin{itemize}
 \item $t$ copies of $\UU\succ c\succ \text{others}$
 \item $\nfrac{m}{3}-1$ copies of $\UU\succ a\succ c\succ \text{others}$
 \item $\nfrac{m}{3}+1$ copies of $\DD\succ b\succ \text{others}$
\end{itemize}
The total number of voters, including the manipulator, is $2t+\nfrac{2m}{3}+1$. Now we show equivalence of the two instances. 

In the forward direction, suppose we have an exact set cover $\TT\subset \SS$. Let the vote of the manipulator $\vvv$ be $c\succ D\succ \text{others}$. We consider the following extension $\overline{\PP}$ of $\PP$. 

$$\WW \succ S_i\succ c\succ \XX\succ (\UU\setminus S_i)\succ \DD$$

On the other hand, if $S_i\in\SS\setminus\TT$, then we have:
$$\WW \succ \XX\succ S_i\succ c\succ (\UU\setminus S_i)\succ \DD$$

We claim that $c$ is the unique simplified Bucklin winner in the profile $(\overline{\PP},\WW,\vvv)$. Notice that the simplified Bucklin score of $c$ is $m+1$ in this extension, since it appears in the top $m+1$ positions in the $m/3$ votes corresponding to the set cover, $t$ votes from the complete profile $\QQ$ and one vote $\vvv$ of the manipulator. For any other candidate $u_i\in \UU$, $u_i$ appears in the top $m+1$ positions once in $\overline{\PP}$ and $t+\frac{m}{3}-1$ times in $\QQ$. Thus, $u_i$ does not get majority in top $m+1$ positions making its simplified Bucklin score at least $m+2$. Hence, $c$ is the unique simplified Bucklin winner in the profile $(\overline{\PP},\WW,\vvv)$.
Similarly, the candidate $w_1$ appears only $t$ times in the top $m+1$ positions. 
The same can be argued for the remaining candidates in $\DD, \WW,$ and $w$. 

In the reverse direction, suppose the \WM is a \YES instance. We may assume without loss of generality that the manipulator's vote $\vvv$ is $c\succ \DD\succ \text{others}$, since the simplified Bucklin score of the candidates in $\DD$ is at least $2m$. Let $\overline{\PP}$ be the extension of $\PP$ such that $c$ is the unique winner in the profile $(\overline{\PP},\QQ,\vvv)$. As every candidate in $\WW$ is ranked within top $m+2$ positions in $t+\frac{m}{3}+1$ votes in $\QQ$, for $c$ to win, $c\succ \XX$ must hold in at least $\frac{m}{3}$ votes in $\overline{\PP}$. In those votes, all the candidates in $S_i$ are also within top $m+2$ positions. Now if any candidate in $\UU$ is within top $m+1$ positions in $\overline{\PP}$ more than once, then $c$ will not be the unique winner. Hence, the $S_i$'s corresponding to the votes where $c\succ \XX$ in $\overline{\PP}$ form an exact set cover. 

The reduction above also works for the Bucklin voting rule. Specifically, the argument for the forward direction is exactly the same as the simplified Bucklin above and the argument for the reverse direction is as follows. Every candidate in $\WW$ is ranked within top $m+2$ positions in $t+\frac{m}{3}+1$ votes in $\QQ$ and $c$ is never placed within top $m+2$ positions in any vote in $\QQ$. Hence, for $c$ to win, $c\succ \XX$ must hold in at least $\frac{m}{3}$ votes in $\overline{\PP}$. In those votes, all the candidates in $S_i$ are also within top $m$ positions. Notice that $c$ never gets placed within top $m$ positions in any vote in $(\overline{\PP},\QQ)$. Now if any candidate $x\in\UU$ is within top $m$ positions in $\overline{\PP}$ more than once, then $x$ gets majority within top $m$ positions and thus $c$ cannot win.

The result for the Fallback and simplified Fallback voting rules follow from the corresponding results for the Bucklin and simplified Bucklin voting rules respectively since every Bucklin and simplified Bucklin election is also a Fallback and simplified Fallback election respectively.
\end{proof}

\subsection{Strong Manipulation} 

We know that the \CM problem is \NPC for the Borda, maximin, and Copeland$^\alpha$ voting rules for every rational $\alpha\in[0,1]\setminus\{0.5\}$, when we have two manipulators. Thus, it follows from~\Cref{cm_hard} that \SM is \NPH{} for Borda, maximin, and Copeland$^\alpha$ voting rules for every rational $\alpha\in[0,1]\setminus\{0.5\}$ for at least two manipulators. 

For the case of one manipulator, \SM turns out to be polynomial-time solvable for most other voting rules. For Copeland$^\alpha$, however, we show that the problem is \coNPH for every $\alpha\in[0,1]$ for a single manipulator, even when the number of undetermined pairs in each vote is bounded by a constant. This is achieved by a careful reduction from $\overline{\text{X3C}}$. The following lemma has been used before 
\cite{mcgarvey1953theorem}.
\begin{lemma}\label{thm:mcgarvey}
 For any function $f:\mathcal{C} \times \mathcal{C} \longrightarrow \mathbb{Z}$, such that
 \begin{enumerate}[leftmargin=0cm,itemindent=.5cm,labelwidth=\itemindent,labelsep=0cm,align=left]
  \item $\forall a,b \in \mathcal{C}, f(a,b) = -f(b,a)$.
  \item $\forall a,b, c, d \in \mathcal{C}, f(a,b) + f(c,d)$ is even,
 \end{enumerate}
 there exists a $n$-voters' profile such that for all $a,b \in \mathcal{C}$, $a$ defeats 
 $b$ with a margin of $f(a,b)$. Moreover, 
 $$n \text{ is even and } n = O\left(\sum_{\{a,b\}\in \mathcal{C}\times\mathcal{C}} |f(a,b)|\right)$$
\end{lemma}

We have following intractability result for the \textsc{Strong Manipulation} problem for the Copeland$^\alpha$ rule with one manipulator and at most $10$ undetermined pairs per vote.

\begin{theorem}\label{sm_copeland_hard}
\textsc{Strong Manipulation} is \coNPH for Copeland$^\alpha$ voting rule for every $\alpha\in[0,1]$ even when we have only one manipulator and the number of undetermined pairs in each vote is no more than $10$. 
\end{theorem}

\begin{proof}
 We reduce X3C to \SM for Copeland$^\alpha$ rule. Let $(\mathcal{U} = \{u_1, \ldots, u_m\}, \mathcal{S}=\{S_1,S_2, \dots, S_t\})$ is an X3C instance. We assume, without loss of generality, $t$ to be an even integer (if not, replicate any set from $\mathcal{S}$). We construct a corresponding \WM instance for Copeland$^\alpha$ as follows.
 $$ \text{Candidate set } \mathcal{C} = \mathcal{U} \cup \{c, w, z, d\} $$
 Partial votes $\mathcal{P}$:
 $$ \forall i\le t, (\mathcal{U}\setminus S_i) \succ c\succ z\succ d\succ S_i\succ w \setminus \{ \{z, c\} \times (S_i\cup \{d,w\}) \}$$
 Notice that the number of undetermined pairs in every vote in $\PP$ is $10$. Now we add a set $\QQ$ of complete votes with $|\QQ|$ even and $|\QQ|=poly(m,t)$ using \Cref{thm:mcgarvey} to achieve the following margin of victories in pairwise elections. \Cref{fig:sm_copeland_hard} shows the weighted majority graph of the resulting election.
 
 \begin{itemize}
  \item $D_\QQ(d,z) = D_\QQ(z,c) = D_\QQ(c,d) = D_\QQ(w,z) = 4t$
  \item $D_\QQ(u_i,d) = D_\QQ(c,u_i) = 4t ~\forall u_i\in \mathcal{U}$
  \item $D_\QQ(z,u_i) = t ~\forall u_i\in \mathcal{U}$
  \item $D_\QQ(c,w) = t-\frac{2m}{3}-2$
  \item $D_\QQ(u_i, u_{i+1 \pmod* m}) = 4t ~\forall u_i\in \mathcal{U}$
  \item $D_\QQ(a,b)=0$ for every $a,b\in\CC, a\ne b,$ not mentioned above
 \end{itemize}
 
 \begin{figure}[!htbp]
  \begin{center}
   \begin{tikzpicture}
    
  \node[draw,circle] (w) at (10,5) {w};
  \node[draw,circle] (z) at (10,0) {z};
  \node[draw,circle] (c) at (0,0) {c};
  \node[draw,circle] (d) at (0,5) {d};
  \node[draw,circle] (u) at (5,5) {$u_j$};
  \node[draw,circle] (p) at (8,5) {$u_{j+1}$};  
  
  \draw[->] (d) -- node[above,pos=.2] {$4t$} (z);
  \draw[->] (z) -- node[below] {$4t$} (c);
  \draw[->] (c) -- node[left] {$4t$} (d);
  \draw[->] (w) -- node[right] {$4t$} (z);
  \draw[->] (u) -- node[above] {$4t$} (d);
  \draw[->] (c) -- node[above=.2cm] {$4t$} (u);
  \draw[->] (z) -- node[above] {$t$} (u);
  \draw[->] (c) -- node[right=.2cm,pos=.2] {$t-\frac{2m}{3}-2$} (w);
  \draw[->] (u) -- node[above] {$4t$} (p);
   \end{tikzpicture}
  \end{center}
  \caption{Weighted majority graph of the reduced instance in \Cref{sm_copeland_hard}. The weight of all the edges not shown in the figure are $0$. For simplicity, we do not show edges among $\{u_1, \ldots, u_m\}$.}\label{fig:sm_copeland_hard}
 \end{figure}
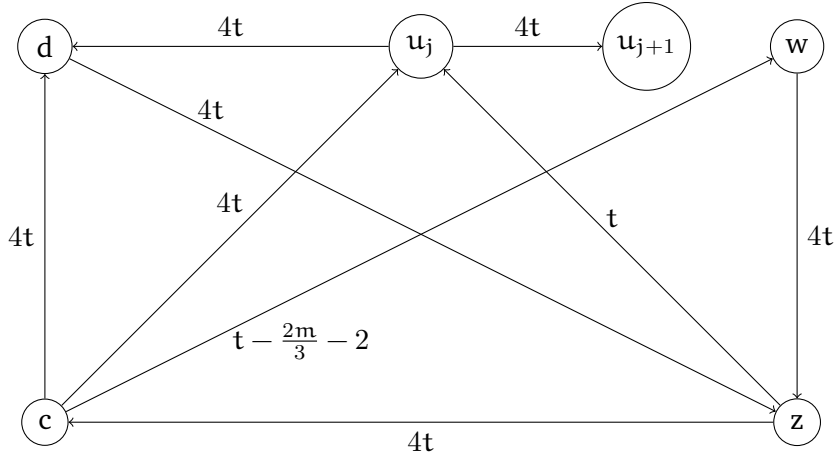

 We have only one manipulator who tries to make $c$ winner. Notice that the number of votes in the \SM instance $(\PP\cup\QQ,1,c)$ including the manipulator's vote is odd (since $|\PP|$ and $|\QQ|$ are even integers). Therefore, $D_{\PP^*\cup\QQ\cup\{v^*\}}(a,b)$ is never zero for every $a,b\in\CC, a\ne b$ in every extension $\PP^*$ of $\PP$ and manipulators vote $v^*$ and consequently the particular value of $\alpha$ does not play any role in this reduction. Hence, we assume, without loss of generality, $\alpha$ to be zero from here on and simply use the term Copeland instead of Copeland$^\alpha$.
 
 Now we show that the X3C instance $(\mathcal{U},\mathcal{S})$ is a \YES instance if and only if the \SM instance $(\PP\cup\QQ,1,c)$ is a \NO instance (a \SM instance is a \NO instance if there does not exist a vote of the manipulator which makes $c$ the unique winner in every extension of the partial votes). We can assume without loss of generality that manipulator puts $c$ at first position and $z$ at last position in her vote $\vvv$.
  
 Assume that the X3C instance is a \YES instance. Suppose (by renaming) that $S_1, \dots, S_{\frac{m}{3}}$ forms an exact set cover. We claim that the following extension $\overline{\PP}$ of $\PP$ makes both $z$ and $c$ Copeland co-winners.
 
 Extension $\overline{\PP}$ of $\mathcal{P}$:
 $$ i\le \frac{m}{3}, (\mathcal{U}\setminus S_i) \succ c\succ z\succ d\succ S_i\succ w $$
 $$ i\ge \frac{m}{3} + 1, (\mathcal{U}\setminus S_i) \succ d\succ S_i\succ w\succ c\succ z $$
 
 We have summarize the pairwise margins between $z$ and $c$ and the rest of the candidates from the profile  $(\overline{\PP}\cup\QQ\cup\vvv)$ in \Cref{tbl:copeland}. The candidates $z$ and $c$ are the co-winners with Copeland score $(m+1)$.
 
 \begin{table}[htbp!]
  \centering
  {
   \renewcommand*{\arraystretch}{1.5}
  \begin{tabular}{|c|c|c|c|c|}\cline{1-2}\cline{4-5}
   $\CC\setminus\{z\}$ & $D_{\overline{\PP}\cup\QQ\cup\vvv}(z, \cdot)$ && $\CC\setminus\{c\}$ & $D_{\overline{\PP}\cup\QQ\cup\vvv}(c, \cdot)$ \\\cline{1-2}\cline{4-5}
   
   $c$ & $\ge 3t$ && $z, u_i\in\UU$ & $\ge 3t$ \\\cline{1-2}\cline{4-5}
   $w, d$ & $\le -3t$ && $w$ & $-1$ \\\cline{1-2}\cline{4-5}
   $u_i\in\UU$ & $1$ && $d$ & $\le -3t$ \\\cline{1-2}\cline{4-5}
  \end{tabular}}
  \caption{$D_{\overline{\PP}\cup\QQ\cup\vvv}(z, \cdot)$ and $D_{\overline{\PP}\cup\QQ\cup\vvv}(c, \cdot)$}\label{tbl:copeland}
 \end{table}
 
 For the other direction, notice that Copeland score of $c$ is at least $m+1$ since $c$ defeats $d$ and every candidate in $\mathcal{U}$ in every extension of $\PP$. Also notice that the Copeland score of $z$ can be at most $m+1$ since $z$ loses to $w$ and $d$ in every extension of $\PP$. Hence the only way $c$ cannot be the unique winner is that $z$ defeats all candidates in $\mathcal{U}$ and $w$ defeats $c$. 
 
 This requires $w\succ c$ in at least $t-\frac{m}{3}$ extensions of $\mathcal{P}$. We claim that the sets $S_i$ in the remaining of the extensions where $c\succ w$ forms an exact set cover for $(\UU,\SS)$. Indeed, otherwise some candidate $u_i\in \mathcal{U}$ is not covered. Then, notice that $u_i \succ z$ in all $t$ votes, making $D(z,u_i) = -1$.
\end{proof}

\subsection{Opportunistic Manipulation}

All our reductions for the \coNPH{}ness for \OM start from $\overline{\text{X3C}}$. We note that all our hardness results  hold even when there is only one manipulator.  Our overall approach is the following. We engineer a set of partial votes in such a way that the manipulator is forced to vote in a limited number of ways to have any hope of making her favorite candidate win. For each such vote, we demonstrate a viable extension where the vote fails to make the candidate a winner, leading to a \NO{} instance of \OM. These extensions rely on the existence of an exact cover. On the other hand, we show that if there is no exact set cover, then there is no viable extension, thereby leading to an instance that is vacuously a \YES{} instance of \OM. Our first result on \OM shows that the \OM problem is \coNPH for the $k$-approval voting rule for constant $k\ge 3$ even when the number of manipulators is one and the number of undetermined pairs in each vote is no more than $15$.

\begin{theorem}\label{thm:kappOM}
 The \OM problem is \coNPH for the $k$-approval voting rule for constant $k\ge 3$ even when the number of manipulators is one and the number of undetermined pairs in each vote is no more than $15$.
\end{theorem}

\begin{proof}
 We reduce $\overline{\text{X3C}}$ to \OM for $k$-approval rule. Let $(\mathcal{U} = \{u_1, \ldots, u_m\}, \mathcal{S}=\{S_1,S_2, \dots, S_t\})$ is an $\overline{\text{X3C}}$ instance. We construct a corresponding \OM instance for $k$-approval voting rule as follows. We begin by introducing a candidate for every element of the universe, along with $k-3$ dummy candidates (denoted by $\WW$), and special candidates $\{c,z_1,z_2, d,x,y\}$. Formally, we have:
 $$ \text{Candidate set } \mathcal{C} = \mathcal{U} \cup \{c, z_1, z_2, d, x, y\} \cup \WW.$$ 
Now, for every set $S_i$ in the universe, we define the following total order on the candidate set, which we denote by $\PP^\prime_i$: 
 $$\WW\succ S_i \succ y \succ z_1 \succ z_2 \succ x \succ (\mathcal{U}\setminus S_i) \succ c \succ d$$ 
 Using $\PP^\prime_i$, we define the partial vote $\PP_i$ as follows: $$\PP_i = \PP^\prime_i \setminus (\{ \{y, x, z_1, z_2\} \times S_i \} \cup \{(z_1, z_2), (x, z_1), (x,z_2)\}).$$ 
 
 We denote the set of partial votes $\{\PP_i: i\in[t]\}$ by $\PP$ and $\{\PP^\prime_i: i\in[t]\}$ by $\PP^\prime$. We remark that the number of undetermined pairs in each partial vote $\PP_i$ is $15$. 
 
 We now invoke Lemma 1 from~\cite{journalsDeyMN16}, which allows to achieve any pre-defined scores on the candidates using only polynomially many additional votes. Using this, we add a set $\QQ$ of complete votes with $|\QQ|=\text{poly}(m,t)$ to ensure the following scores, where we denote the $k$-approval score of a candidate from a set of votes $\VV$ by $s_\VV(\cdot)$: $s_\QQ(z_1) = s_\QQ(z_2) =  s_\QQ(y) = s_\QQ(c) - \nfrac{m}{3}; s_\QQ(d), s_\QQ(w) \le s_\QQ(c) - 2t ~\forall w\in\WW; s_\QQ(x) = s_\QQ(c) -1; s_{\PP^\prime\cup\QQ} (u_j) = s_\QQ(c) + 1 ~\forall j\in[m]$.
 
 Our reduced instance is $(\PP\cup\QQ,1,c)$. The reasoning for this score configuration will be apparent as we argue the equivalence. We first argue that if we had a \YES{} instance of $\overline{\text{X3C}}$ (in other words, there is no exact cover), then we have a \YES{} instance of \OM. It turns out that this will follow from the fact that there are no viable extensions, because, as we will show next, a viable extension implies the existence of an exact set cover.
 
 To this end, first observe that the partial votes are constructed in such a way that $c$ gets no additional score from \textit{any} extension. Assuming that the manipulator approves $c$ (without loss of generality), the final score of $c$ in any extension is going to be  $s_\QQ(c) + 1$. Now, in any viable extension, every candidate $u_j$ has to be ``pushed out'' of the top $k$ positions at least once. Observe that whenever this happens, $y$ is forced into the top $k$ positions. Since $y$ is behind the score of $c$ by only $m/3$ votes, $S_i$'s can be pushed out of place in only $m/3$ votes. For every $u_j$ to lose one point, these votes must correspond to an exact cover. Therefore, if there is no exact cover, then there is no viable extension, showing one direction of the reduction. 
 
 On the other hand, suppose we have a \NO{} instance of $\overline{\text{X3C}}$ -- that is, there is an exact cover. Let $\{S_i: i\in[\nfrac{m}{3}]\}$ forms an exact cover of $\UU$. We will now use the exact cover to come up with two viable extensions, both of which require the manipulator to vote in different ways to make $c$ win. Therefore, there is no single manipulative vote that accounts for both extensions, leading us to a \NO{} instance of \OM. 

 First, consider this completion of the partial votes: 
 $$ i=1, \WW\succ y \succ x \succ z_1 \succ z_2 \succ S_i \succ (\mathcal{U}\setminus S_i) \succ c \succ d$$
 $$ 2\le i\le \nfrac{m}{3}, \WW\succ y \succ z_1 \succ z_2 \succ x \succ S_i \succ (\mathcal{U}\setminus S_i) \succ c \succ d$$
 $$ \nfrac{m}{3} + 1\le i\le t, \WW\succ S_i \succ y \succ z_1 \succ z_2 \succ x \succ (\mathcal{U}\setminus S_i) \succ c \succ d$$
Notice that in this completion, once accounted for along with the votes in $\QQ$, the score of $c$ is tied with the scores of all $u_j$'s, $z_1, x$ and $y$, while the score of $z_2$ is one less than the score of $c$. Therefore, the only $k$ candidates that the manipulator can afford to approve are $\WW$, the candidates $c,d$ and $z_2$. However, consider the extension that is identical to the above except with the first vote changed to:
 $$ \WW\succ y \succ x \succ z_2 \succ z_1 \succ S_i \succ (\mathcal{U}\setminus S_i) \succ c \succ d$$
Here, on the other hand, the only way for $c$ to be an unique winner is if the manipulator approves $\WW, c,d$ and $z_1$. Therefore, it is clear that there is no way for the manipulator to provide a consolidated vote for both these profiles. Therefore, we have a \NO{} instance of \OM.
\end{proof}

We next move on to the $k$-veto voting rule and show that the \OM problem for the $k$-veto is \coNPH for every constant $k\ge 4$ even when the number of manipulators is one and the number of undetermined pairs in each vote is no more than $15$.

\begin{theorem}\label{thm:kvetoOM}
 The \OM problem is \coNPH for the $k$-veto voting rule for every constant $k\ge 4$ even when the number of manipulators is one and the number of undetermined pairs in each vote is no more than $15$.
\end{theorem}

\begin{proof}
 We reduce X3C to \OM for $k$-veto rule. Let $(\mathcal{U} = \{u_1, \ldots, u_m\}, \mathcal{S}=\{S_1,S_2, \dots, S_t\})$ is an X3C instance. We construct a corresponding \OM instance for $k$-veto voting rule as follows.
 $$ \text{Candidate set } \mathcal{C} = \mathcal{U} \cup \{c, z_1, z_2, d, x, y\} \cup \AA \cup \WW, \text{ where } \AA = \{a_1, a_2, a_3\}, |\WW|=k-4 $$
 For every $i\in[t]$, we define $\mathcal{P}^\prime_i$ as follows:
 $$ \forall i\le t, c \succ \AA \succ (\mathcal{U}\setminus S_i) \succ d \succ S_i \succ y \succ x \succ z_1 \succ z_2\succ \WW$$
 Using $\PP^\prime_i$, we define partial vote $\PP_i = \PP^\prime_i \setminus (\{ \{y, x, z_1, z_2\} \times S_i \} \cup \{(z_1, z_2), (x, z_1), (x,z_2)\})$ for every $i\in[t]$. We denote the set of partial votes $\{\PP_i: i\in[t]\}$ by $\PP$ and $\{\PP^\prime_i: i\in[t]\}$ by $\PP^\prime$. We note that the number of undetermined pairs in each partial vote $\PP_i$ is $15$. Using \Cref{score_gen}, we add a set $\QQ$ of complete votes with $|\QQ|=\text{poly}(m,t)$ to ensure the following. We denote the $k$-veto score of a candidate from a set of votes $\WW$ by $s_\WW(\cdot)$.
 \begin{itemize}
  \item $s_{\PP^\prime\cup\QQ} (z_1) = s_{\PP^\prime\cup\QQ} (z_2) = s_{\PP^\prime\cup\QQ} (c) - \nfrac{m}{3}$
  \item $s_{\PP^\prime\cup\QQ} (a_i) = s_{\PP^\prime\cup\QQ} (u_j) = s_{\PP^\prime\cup\QQ} (w) = s_{\PP^\prime\cup\QQ} (c) ~\forall a_i\in \AA, u_j\in\UU, w\in\WW$
  \item $s_{\PP^\prime\cup\QQ} (y) = s_{\PP^\prime\cup\QQ} (c) - \nfrac{m}{3} - 1$
  \item $s_{\PP^\prime\cup\QQ} (x) = s_{\PP^\prime\cup\QQ} (c) - 2$
 \end{itemize}
 
 We have only one manipulator who tries to make $c$ winner. Now we show that the X3C instance $(\mathcal{U},\mathcal{S})$ is a \YES instance if and only if the \OM instance $(\PP\cup\QQ,1,c)$ is a \NO instance.
 
 In the forward direction, let us now assume that the X3C instance is a \YES instance. Suppose (by renaming) that $S_1, \dots, S_{\nfrac{m}{3}}$ forms an exact set cover. Let us assume that the manipulator's vote $\vvv$ disapproves every candidate in $\WW\cup\AA$ since otherwise $c$ can never win uniquely. We now show that if $\vvv$ does not disapprove $z_1$ then, $\vvv$ is not a $c$-optimal vote. Suppose $\vvv$ does not disapprove $z_1$. Then we consider the following extension $\overline{\PP}$ of $\PP$.
 
 $$ i=1, c \succ \AA \succ (\mathcal{U}\setminus S_i) \succ d \succ y \succ z_1 \succ x \succ z_2 \succ S_i\succ \WW$$
 $$ 2\le i\le \nfrac{m}{3}, c \succ \AA \succ (\mathcal{U}\setminus S_i) \succ d \succ y \succ z_1 \succ z_2 \succ x \succ S_i\succ \WW$$
 $$ \nfrac{m}{3} + 1\le i\le t, c \succ \AA \succ (\mathcal{U}\setminus S_i) \succ d \succ S_i \succ y \succ x \succ z_1 \succ z_2\succ \WW$$
 
 We have the following scores $s_{\overline{\PP}\cup\QQ}(c) = s_{\overline{\PP}\cup\QQ}(z_1) = s_{\overline{\PP}\cup\QQ}(z_2) + 1 = s_{\overline{\PP}\cup\QQ}(x) + 1 = s_{\overline{\PP}\cup\QQ}(u_j) + 1 ~\forall u_j\in\UU$. Hence, both $c$ and $z_1$ win for the votes $\overline{\PP}\cup\QQ\cup\{\vvv\}$. However, the vote $\vvv^\prime$ which disapproves $a_1, a_2, a_3, z_1$ makes $c$ a unique winner for the votes $\overline{\PP}\cup\QQ\cup\{\vvv^\prime\}$. Hence, $\vvv$ is not a $c$-optimal vote. Similarly, we can show that if the manipulator's vote does not disapprove $z_2$ then, the vote is not $c$-optimal. Hence, there does not exist any $c$-optimal vote and the \OM instance is a \NO instance.
 
 In the reverse direction, we show that if the X3C instance is a \NO instance, then there does not exist a vote $\vvv$ of the manipulator and an extension $\overline{\PP}$ of $\PP$ such that $c$ is the unique winner for the votes $\overline{\PP}\cup\QQ\cup\{\vvv^\prime\}$ thereby proving that the \OM instance is vacuously \YES (and thus every vote is $c$-optimal). Notice that, there must be at least $\nfrac{m}{3}$ votes $\PP_1$ in $\overline{\PP}$ where the corresponding $S_i$ gets pushed in bottom $k$ positions since $s_{\PP^\prime\cup\QQ} (u_j) = s_{\PP^\prime\cup\QQ} (c) ~\forall a_i\in \AA, u_j\in\UU$. However, in each vote in $\PP_1$, $y$ is placed within top $m-k$ many position and thus we have $|\PP_1|$ is exactly $\nfrac{m}{3}$ since $s_{\PP^\prime\cup\QQ} (y) = s_{\PP^\prime\cup\QQ} (c) - \nfrac{m}{3} - 1$. Now notice that there must be at least one candidate $u\in\UU$ which is not covered by the sets $S_i$s corresponding to the votes $\PP_1$ because the X3C instance is a \NO instance. Hence, $c$ cannot win the election uniquely irrespective of the manipulator's vote. Thus every vote is $c$-optimal and the \OM instance is a \YES instance.
\end{proof}

We show next similar intractability result for the Borda voting rule too with only at most $7$ undetermined pairs per vote.

\begin{theorem}\label{thm:bordaOM}
 The \OM problem is \coNPH for the Borda voting rule even when the number of manipulators is one and the number of undetermined pairs in every vote is no more than $7$.
\end{theorem}

\begin{proof}
 We reduce X3C to \OM for the Borda rule. Let $(\mathcal{U} = \{u_1, \ldots, u_m\}, \mathcal{S}=\{S_1,S_2, \dots, S_t\})$ is an X3C instance. Without loss of generality we assume that $m$ is not divisible by $6$ (if not, then we add three new elements $b_1, b_2, b_3$ to $\UU$ and a set $\{b_1, b_2, b_3\}$ to $\SS$). We construct a corresponding \OM instance for the Borda voting rule as follows.
 $$ \text{Candidate set } \mathcal{C} = \mathcal{U} \cup \{c, z_1, z_2, d, y\} $$
 For every $i\in[t]$, we define $\mathcal{P}^\prime_i$ as follows:
 
 $$ \forall i\le t, y \succ S_i \succ z_1 \succ z_2 \succ (\mathcal{U}\setminus S_i) \succ d \succ c$$
 
 Using $\PP^\prime_i$, we define partial vote $\PP_i = \PP^\prime_i \setminus (\{ (\{y\} \cup S_i) \times \{z_1, z_2\} \} \cup \{(z_1, z_2)\})$ for every $i\in[t]$. We denote the set of partial votes $\{\PP_i: i\in[t]\}$ by $\PP$ and $\{\PP^\prime_i: i\in[t]\}$ by $\PP^\prime$. We note that the number of undetermined pairs in each partial vote $\PP_i$ is $7$. Using \Cref{score_gen}, we add a set $\QQ$ of complete votes with $|\QQ|=\text{poly}(m,t)$ to ensure the following. We denote the Borda score of a candidate from a set of votes $\WW$ by $s_\WW(\cdot)$.
 
 \begin{itemize}
  \item $s_{\PP^\prime\cup\QQ} (y) = s_{\PP^\prime\cup\QQ} (c) + m + \nfrac{m}{3} + 3$
  \item $s_{\PP^\prime\cup\QQ} (z_1) = s_{\PP^\prime\cup\QQ} (c) - 3\lfloor\nfrac{m}{6}\rfloor - 2$
  \item $s_{\PP^\prime\cup\QQ} (z_2) = s_{\PP^\prime\cup\QQ} (c) - 5\lfloor\nfrac{m}{6}\rfloor - 3$
  \item $s_{\PP^\prime\cup\QQ} (u_i) = s_{\PP^\prime\cup\QQ} (c) + m + 5 - i ~\forall i\in[m]$
  \item $s_{\PP^\prime\cup\QQ} (d) \le s_{\PP^\prime\cup\QQ} (c) - 5m$
 \end{itemize}
 
 We have only one manipulator who tries to make $c$ winner. Now we show that the X3C instance $(\mathcal{U},\mathcal{S})$ is a \YES instance if and only if the \OM instance $(\PP\cup\QQ,1,c)$ is a \NO instance. Notice that we can assume without loss of generality that the manipulator places $c$ at the first position, $d$ at the second position, the candidate $u_i$ at $(m+5-i)^{th}$ position for every $i\in[m]$, and $y$ at the last position, since otherwise $c$ can never win uniquely irrespective of the extension of $\PP$ (that it, the manipulator's vote looks like $c\succ d \succ \{z_1, z_2\} \succ u_m \succ u_{m-1} \succ \cdots \succ u_1 \succ y$). 
 
 In the forward direction, let us now assume that the X3C instance is a \YES instance. Suppose (by renaming) that $S_1, \dots, S_{\nfrac{m}{3}}$ forms an exact set cover. Let the manipulator's vote $\vvv$ be $c \succ d \succ z_1 \succ z_2 \succ u_m \succ \cdots \succ u_1 \succ y$. We now argue that $\vvv$ is not a $c$-optimal vote. The other case where the manipulator's vote $\vvv^\prime$ be $c \succ d \succ z_2 \succ z_1 \succ u_m \succ \cdots \succ u_1 \succ y$ can be argued similarly. We consider the following extension $\overline{\PP}$ of $\PP$.
 
 $$ 1\le i\le \lfloor\nfrac{m}{6}\rfloor, z_2 \succ y \succ S_i \succ z_1 \succ (\mathcal{U}\setminus S_i) \succ d \succ c $$
 $$ \lceil\nfrac{m}{6}\rceil\le i\le \nfrac{m}{3}, z_1 \succ y \succ S_i \succ z_2 \succ (\mathcal{U}\setminus S_i) \succ d \succ c $$
 $$ \nfrac{m}{3} + 1\le i\le t, y \succ S_i \succ z_1 \succ z_2 \succ (\mathcal{U}\setminus S_i) \succ d \succ c$$
 
 We have the following Borda scores $s_{\overline{\PP}\cup\QQ\cup\{\vvv\}}(c) = s_{\overline{\PP}\cup\QQ\cup\{\vvv\}}(y) + 1 = s_{\overline{\PP}\cup\QQ\cup\{\vvv\}}(z_2) + 6 = s_{\overline{\PP}\cup\QQ\cup\{\vvv\}}(z_1) = s_{\overline{\PP}\cup\QQ\cup\{\vvv\}}(u_i) + 1 ~\forall i\in[m]$. Hence, $c$ does not win uniquely for the votes $\overline{\PP}\cup\QQ\cup\{\vvv\}$. However, $c$ is the unique winner for the votes $\overline{\PP}\cup\QQ\cup\{\vvv^\prime\}$. Hence, there does not exist any $c$-optimal vote and the \OM instance is a \NO instance.

 In the reverse direction, we show that if the X3C instance is a \NO instance, then there does not exist a vote $\vvv$ of the manipulator and an extension $\overline{\PP}$ of $\PP$ such that $c$ is the unique winner for the votes $\overline{\PP}\cup\QQ\cup\{\vvv^\prime\}$ thereby proving that the \OM instance is vacuously \YES (and thus every vote is $c$-optimal). Notice that the score of $y$ must decrease by at least $\nfrac{m}{3}$ for $c$ to win uniquely. However, in every vote $v$ where the score of $y$ decreases by at least one in any extension $\overline{\PP}$ of $\PP$, at least one of $z_1$ or $z_2$ must be placed at top position of the vote $v$. However, the candidates $z_1$ and $z_2$ can be placed at top positions of the votes in $\overline{\PP}$ at most $\nfrac{m}{3}$ many times while ensuring $c$ does not lose the election. Also, even after manipulator places the candidate $u_i$ at $(m+5-i)^{th}$ position for every $i\in[m]$, for $c$ to win uniquely, the score of every $u_i$ must decrease by at least one. Hence, altogether, there will be exactly $\nfrac{m}{3}$ votes (denoted by the set $\PP_1$) in any extension of $\PP$ where $y$ is placed at the second position. However, since the X3C instance is a \NO instance, the $S_i$s corresponding to the votes in $\PP_1$ does not form a set cover. Let $u\in\UU$ be an element not covered by the $S_i$s corresponding to the votes in $\PP_1$. Notice that the score of $u$ does not decrease in the extension $\overline{\PP}$ and thus $c$ does not win uniquely irrespective of the manipulator's vote. Thus every vote is $c$-optimal and thus the \OM instance is a \YES instance. Thus every vote is $c$-optimal and the \OM instance is a \YES instance.
\end{proof}

For the maximin voting rule, we show intractability of \OM with one manipulator even when the number of undetermined pairs in every vote is at most $8$.

\begin{theorem}\label{thm:maximinOM}
 The \OM problem is \coNPH for the maximin voting rule even when the number of manipulators is one and the number of undetermined pairs in every vote is no more than $8$.
\end{theorem}

\begin{proof}
 We reduce X3C to \OM for the maximin rule. Let $(\mathcal{U} = \{u_1, \ldots, u_m\}, \mathcal{S}=\{S_1,S_2, \dots, S_t\})$ is an X3C instance. We construct a corresponding \OM instance for the maximin voting rule as follows.
 $$ \text{Candidate set } \mathcal{C} = \mathcal{U} \cup \{c, z_1, z_2, z_3, d, x, y\} $$
 For every $i\in[t]$, we define $\mathcal{P}^\prime_i$ as follows:
 
 $$ \forall i\le t, S_i \succ x \succ d \succ y \succ (\mathcal{U}\setminus S_i) \succ z_1 \succ z_2 \succ z_3$$
 
 Using $\PP^\prime_i$, we define partial vote $\PP_i = \PP^\prime_i \setminus (\{ (\{x\} \cup S_i) \times \{d, y\} \})$ for every $i\in[t]$. We denote the set of partial votes $\{\PP_i: i\in[t]\}$ by $\PP$ and $\{\PP^\prime_i: i\in[t]\}$ by $\PP^\prime$. We note that the number of undetermined pairs in each partial vote $\PP_i$ is $8$. We define another partial vote $\ppp$ as follows.
 
 $$ \ppp = (z_1 \succ z_2 \succ z_3 \succ \text{ others }) \setminus \{(z_1, z_2), (z_2, z_3), (z_1, z_3)\} $$
 
 Using \Cref{thm:mcgarvey}, we add a set $\QQ$ of complete votes with $|\QQ|=\text{poly}(m,t)$ to ensure the following pairwise margins (notice that the pairwise margins among $z_1, z_2,$ and $z_3$ does not include the partial vote $\ppp$). \Cref{fig:maximinOM} shows the weighted majority graph of the resulting election.
 
 \begin{itemize}
  \item $D_{\PP^\prime\cup\QQ\cup\{\ppp\}} (d, c) = 4t + 1$
  \item $D_{\PP^\prime\cup\QQ\cup\{\ppp\}} (x, d) = 4t + \nfrac{2m}{3} + 1$
  \item $D_{\PP^\prime\cup\QQ\cup\{\ppp\}} (y, x) = 4t - \nfrac{2m}{3} + 1$
  \item $D_{\PP^\prime\cup\QQ\cup\{\ppp\}} (d, u_j) = 4t - 1 ~\forall u_j\in\UU$
  \item $D_{\PP^\prime\cup\QQ} (z_1, z_2) = D_{\PP^\prime\cup\QQ} (z_2, z_3) = D_{\PP^\prime\cup\QQ} (z_3, z_1) = 4t + 2$
  \item $|D_{\PP^\prime\cup\QQ\cup\{\ppp\}} (a, b)| \le 1$ for every $a, b\in \CC$ not defined above.
 \end{itemize}

 \begin{figure}[!htbp]
  \begin{center}
   \begin{tikzpicture}
    
  \node[draw,circle] (x) at (5,0) {x};
  \node[draw,circle] (y) at (10,0) {y};
  \node[draw,circle] (c) at (0,0) {c};
  \node[draw,circle] (d) at (0,5) {d};
  \node[draw,circle] (u) at (5,5) {$u_j$};
  \node[draw,circle] (z1) at (10,5) {$z_1$};
  \node[draw,circle] (z2) at (10,2.5) {$z_2$};
  \node[draw,circle] (z3) at (13,5) {$z_3$};
  
  \draw[->] (d) -- node[left] {$4t+1$} (c);
  \draw[->] (x) -- node[left,pos=.3] {$4t+\nfrac{2m}{3}+1$} (d);
  \draw[->] (y) -- node[below] {$4t-\nfrac{2m}{3}+1$} (x);
  \draw[->] (d) -- node[above] {$4t-1$} (u);
  \draw[dashed,->] (z1) -- node[left] {$4t+2$} (z2);
  \draw[dashed,->] (z2) -- node[right] {$4t+2$} (z3);
  \draw[dashed,->] (z3) -- node[above] {$4t+2$} (z1);
   \end{tikzpicture}
  \end{center}
  \caption{Weighted majority graph of the reduced instance in \Cref{thm:maximinOM}. Solid line and dashed line represent pairwise margins in $\PP^\prime\cup\QQ\cup\{\ppp\}$ and $\PP^\prime\cup\QQ$ respectively. The weight of all the edges not shown in the figure are within $-1$ to $1$. For simplicity, we do not show edges among $\{u_1, \ldots, u_m\}$.}\label{fig:maximinOM}
 \end{figure}
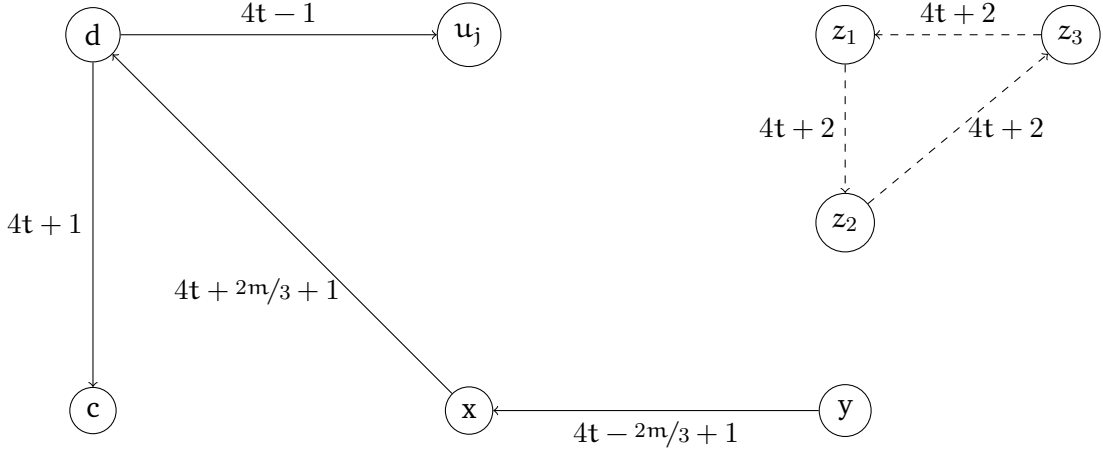

 We have only one manipulator who tries to make $c$ winner. Now we show that the X3C instance $(\mathcal{U},\mathcal{S})$ is a \YES instance if and only if the \OM instance $(\PP\cup\QQ\cup\{\ppp\},1,c)$ is a \NO instance. Notice that we can assume without loss of generality that the manipulator's vote prefers $c$ to every other candidate, $y$ to $x$, $x$ to $d$, and $d$ to $u_j$ for every $u_j\in\UU$.
 
 In the forward direction, let us now assume that the X3C instance is a \YES instance. Suppose (by renaming) that $S_1, \dots, S_{\nfrac{m}{3}}$ forms an exact set cover. Notice that the manipulator's vote  must prefer either $z_2$ to $z_1$ or $z_1$ to $z_3$ or $z_3$ to $z_2$. We show that if the manipulator's vote $\vvv$ prefers $z_2$ to $z_1$, then $\vvv$ is not a $c$-optimal vote. The other two cases are symmetrical. Consider the following extension $\overline{\PP}$ of $\PP$ and $\overline{\ppp}$ of $\ppp$.
 
 $$ 1\le i\le \nfrac{m}{3}, d \succ y \succ S_i \succ x \succ (\mathcal{U}\setminus S_i) \succ z_1 \succ z_2 \succ z_3 $$
 $$ \nfrac{m}{3} + 1 \le i\le t, S_i \succ x \succ d \succ y \succ (\mathcal{U}\setminus S_i) \succ z_1 \succ z_2 \succ z_3 $$
 $$ \overline{\ppp} = z_2 \succ z_3 \succ z_1 \succ \text{ others } $$
 
 From the votes in $\overline{\PP}\cup\QQ\cup\{\vvv, \overline{\ppp}\}$, the maximin score of $c$ is $-4t$, of $d, x, u_j ~\forall u_j\in\UU$ are $-4t-2$, of $z_1, z_3$ are at most than $-4t-2$, and of $z_2$ is $-4t$. Hence, $c$ is not the unique maximn winner. However, the manipulator's vote $c \succ z_1 \succ z_2 \succ z_3 \succ \text{ other }$ makes $c$ the unique maximin winner. Hence, $\vvv$ is not a $c$-optimal vote.
 
 For the reverse direction, we show that if the X3C instance is a \NO instance, then there does not exist a vote $\vvv$ of the manipulator and an extension $\overline{\PP}$ of $\PP$ such that $c$ is the unique winner for the votes $\overline{\PP}\cup\QQ\cup\{\vvv^\prime\}$ thereby proving that the \OM instance is vacuously \YES (and thus every vote is $c$-optimal). Consider any extension $\overline{\PP}$ of $\PP$. Notice that, for $c$ to win uniquely, $y \succ x$ must be at least $\nfrac{m}{3}$ of the votes in $\overline{\PP}$; call these set of votes $\PP_1$. However, $d\succ x$ in every vote in $\PP_1$ and $d\succ x$ can be in at most $\nfrac{m}{3}$ votes in $\overline{\PP}$ for $c$ to win uniquely. Hence, we have $|\PP_1| = \nfrac{m}{3}$. Also for $c$ to win, each $d\succ u_j$ must be at least one vote of $\overline{\PP}$ and $d\succ u_j$ is possible only in the votes in $\PP_1$. However, the sets $S_i$s corresponding to the votes in $\PP_1$ does not form a set cover since the X3C instance is a \NO instance. Hence, there must exist a $u_j\in \UU$ for which $u_j\succ d$ in every vote in $\overline{\PP}$ and thus $c$ cannot win uniquely irrespective of the vote of the manipulator. Thus every vote is $c$-optimal and the \OM instance is a \YES instance.
\end{proof}

Our next result proves that the \OM problem is \coNPH for the Copeland$^\alpha$ voting rule too for every $\alpha\in[0,1]$ even with one manipulator and at most $8$ undetermined pairs per vote.

\begin{theorem}\label{thm:copelandOM}
 The \OM problem is \coNPH for the Copeland$^\alpha$ voting rule for every $\alpha\in[0,1]$ even when the number of manipulators is one and the number of undetermined pairs in each vote is no more than $8$.
\end{theorem}

\begin{proof}
 We reduce X3C to \OM for the Copeland$^\alpha$ voting rule. Let $(\mathcal{U} = \{u_1, \ldots, u_m\}, \mathcal{S}=\{S_1,S_2, \dots, S_t\})$ is an X3C instance. We construct a corresponding \OM instance for the Copeland$^\alpha$ voting rule as follows.
 $$ \text{Candidate set } \mathcal{C} = \mathcal{U} \cup \{c, z_1, z_2, z_3, d_1, d_2, d_3, x, y\} $$
 For every $i\in[t]$, we define $\mathcal{P}^\prime_i$ as follows:
 
 $$ \forall i\le t, S_i \succ x \succ y \succ c \succ \text{ others}$$
 
 Using $\PP^\prime_i$, we define partial vote $\PP_i = \PP^\prime_i \setminus (\{ (\{x\} \cup S_i) \times \{c, y\} \})$ for every $i\in[t]$. We denote the set of partial votes $\{\PP_i: i\in[t]\}$ by $\PP$ and $\{\PP^\prime_i: i\in[t]\}$ by $\PP^\prime$. We note that the number of undetermined pairs in each partial vote $\PP_i$ is $8$. We define another partial vote $\ppp$ as follows.
 
 $$ \ppp = (z_1 \succ z_2 \succ z_3 \succ \text{ others }) \setminus \{(z_1, z_2), (z_2, z_3), (z_1, z_3)\} $$
 
 Using \Cref{thm:mcgarvey}, we add a set $\QQ$ of complete votes with $|\QQ|=\text{poly}(m,t)$ to ensure the following pairwise margins (notice that the pairwise margins among $z_1, z_2,$ and $z_3$ does not include the partial vote $\ppp$). \Cref{fig:copelandOM} shows the weighted majority graph of the resulting election.
 
 \begin{itemize}
  \item $D_{\PP^\prime\cup\QQ\cup\{\ppp\}} (u_j, c) = 2 ~\forall u_j\in\UU$
  \item $D_{\PP^\prime\cup\QQ\cup\{\ppp\}} (x, y) = \nfrac{2m}{3}$
  \item $D_{\PP^\prime\cup\QQ\cup\{\ppp\}} (c, y) = D_{\PP^\prime\cup\QQ\cup\{\ppp\}} (x, c) = D_{\PP^\prime\cup\QQ\cup\{\ppp\}} (d_i, c) = D_{\PP^\prime\cup\QQ\cup\{\ppp\}} (z_k,c) = D_{\PP^\prime\cup\QQ\cup\{\ppp\}} (u_j, x) = D_{\PP^\prime\cup\QQ\cup\{\ppp\}} (x, z_k) = D_{\PP^\prime\cup\QQ\cup\{\ppp\}} (d_i, x) = D_{\PP^\prime\cup\QQ\cup\{\ppp\}} (y, u_j) = D_{\PP^\prime\cup\QQ\cup\{\ppp\}} (d_i, y) = D_{\PP^\prime\cup\QQ\cup\{\ppp\}} (y, z_k) = D_{\PP^\prime\cup\QQ\cup\{\ppp\}} (z_k, u_j) = D_{\PP^\prime\cup\QQ\cup\{\ppp\}} (u_j, d_i) = D_{\PP^\prime\cup\QQ\cup\{\ppp\}} (z_k, d_1) = D_{\PP^\prime\cup\QQ\cup\{\ppp\}} (z_k, d_2) = D_{\PP^\prime\cup\QQ\cup\{\ppp\}} (d_3, z_k) =  4t ~\forall i,k\in[3], j\in[m]$
  
  \item $D_{\PP^\prime\cup\QQ\cup\{\ppp\}} (u_j, u_\el) = -4t$ for at least $\nfrac{m}{3}$ many $u_\el\in\UU$  
  \item $D_{\PP^\prime\cup\QQ} (z_1, z_2) = D_{\PP^\prime\cup\QQ} (z_2, z_3) = D_{\PP^\prime\cup\QQ} (z_3, z_1) = 1$
  \item $|D_{\PP^\prime\cup\QQ\cup\{\ppp\}} (a, b)| \le 1$ for every $a, b\in \CC$ not defined above.
 \end{itemize}

 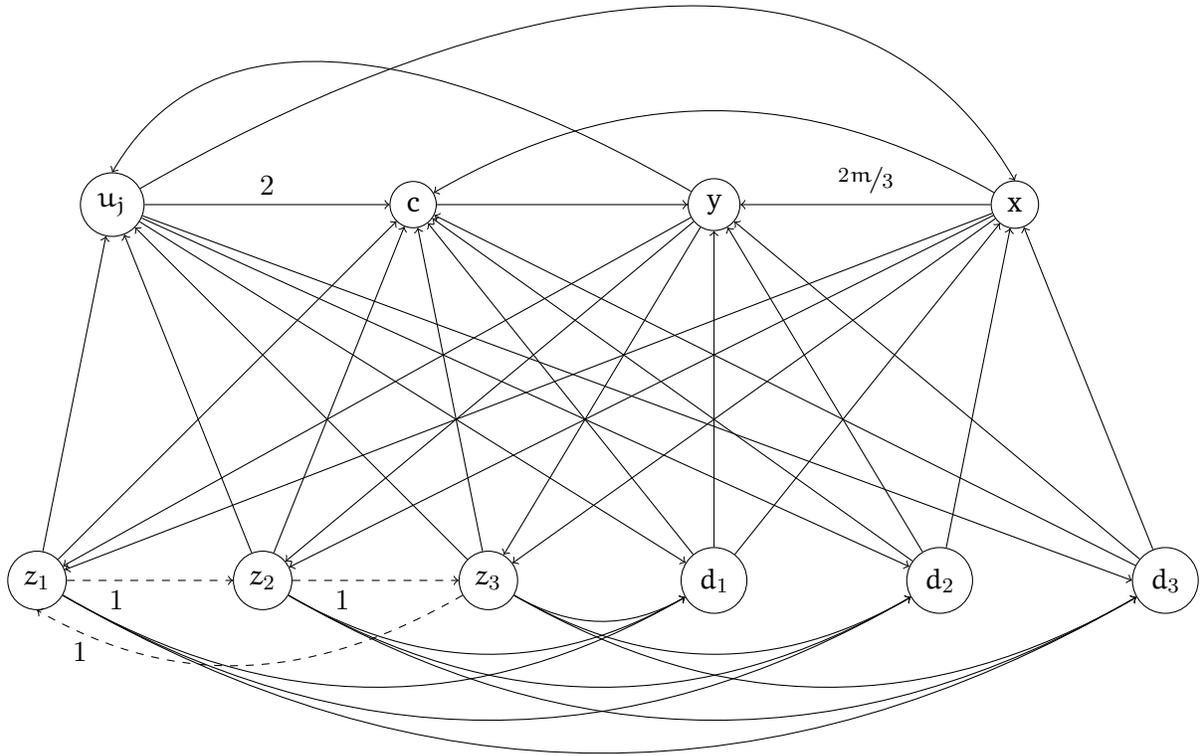
\begin{figure}[!htbp]
  \begin{center}
   \begin{tikzpicture}
    
  \node[draw,circle] (x) at (10,5) {x};
  \node[draw,circle] (y) at (6,5) {y};
  \node[draw,circle] (c) at (2,5) {c};
  \node[draw,circle] (u) at (-2,5) {$u_j$};
  
  \node[draw,circle] (z1) at (-3,0) {$z_1$};
  \node[draw,circle] (z2) at (0,0) {$z_2$};
  \node[draw,circle] (z3) at (3,0) {$z_3$};
  
  \node[draw,circle] (d1) at (6,0) {$d_1$};
  \node[draw,circle] (d2) at (9,0) {$d_2$};
  \node[draw,circle] (d3) at (12,0) {$d_3$};
  
  \draw[->] (u) -- node[above] {$2$} (c);
  \draw[->] (x) -- node[above] {$\nfrac{2m}{3}$} (y);
  \draw[->] (c) -- node[above] {} (y);
  \draw[->] (x) to [out=150,in=30] node[above] {} (c);
  
  \draw[->] (d1) -- node[right,pos=.2] {} (c);
  \draw[->] (d2) -- node[right,pos=.2] {} (c);
  \draw[->] (d3) -- node[right,pos=.2] {} (c);
  \draw[->] (z1) -- node[right,pos=.2] {} (c);
  \draw[->] (z2) -- node[right,pos=.2] {} (c);
  \draw[->] (z3) -- node[right,pos=.2] {} (c);
  
  \draw[->] (u) to [out=30,in=120] node[above] {} (x.north);
  \draw[->] (x) -- node[right,pos=.3] {} (z1);
  \draw[->] (x) -- node[right,pos=.3] {} (z2);
  \draw[->] (x) -- node[right,pos=.3] {} (z3);
  \draw[->] (d1) -- node[right,pos=.1] {} (x);
  \draw[->] (d2) -- node[right,pos=.1] {} (x);
  \draw[->] (d3) -- node[right,pos=.1] {} (x);
  \draw[->] (y) to [out=150,in=60] node[above] {} (u.north);
  \draw[->] (d1) -- node[right,pos=.1] {} (y);
  \draw[->] (d2) -- node[right,pos=.1] {} (y);
  \draw[->] (d3) -- node[right,pos=.1] {} (y);
  \draw[->] (y) -- node[right,pos=.3] {} (z1);
  \draw[->] (y) -- node[right,pos=.3] {} (z2);
  \draw[->] (y) -- node[right,pos=.3] {} (z3);
  \draw[->] (z1) -- node[right,pos=.2] {} (u);
  \draw[->] (z2) -- node[right,pos=.2] {} (u);
  \draw[->] (z3) -- node[right,pos=.2] {} (u);
  \draw[->] (u) -- node[right,pos=.3] {} (d1);
  \draw[->] (u) -- node[right,pos=.3] {} (d2);
  \draw[->] (u) -- node[right,pos=.3] {} (d3.west);
  \draw[->] (z1) to [out=-30,in=-150] node[right,pos=.3] {} (d1);
  \draw[->] (z1) to [out=-30,in=-150] node[right,pos=.3] {} (d2);
  \draw[->] (z1) to [out=-30,in=-150] node[right,pos=.3] {} (d3);
  \draw[->] (z2) to [out=-30,in=-150] node[right,pos=.3] {} (d1);
  \draw[->] (z2) to [out=-30,in=-150] node[right,pos=.3] {} (d2);
  \draw[->] (z2) to [out=-30,in=-150] node[right,pos=.3] {} (d3);
  \draw[->] (z3) to [out=-30,in=-150] node[right,pos=.3] {} (d1);
  \draw[->] (z3) to [out=-30,in=-150] node[right,pos=.3] {} (d2);
  \draw[->] (z3) to [out=-30,in=-150] node[right,pos=.3] {} (d3);
  
  \draw[->,dashed] (z1) -- node[below,pos=.3] {$1$} (z2);
  \draw[->,dashed] (z2) -- node[below,pos=.3] {$1$} (z3);
  \draw[->,dashed] (z3) to [out=-150,in=-30] node[below,pos=.9] {$1$} (z1.south);
   \end{tikzpicture}
  \end{center}
  \caption{Weighted majority graph of the reduced instance in \Cref{thm:copelandOM}. Solid line and dashed line represent pairwise margins in $\PP^\prime\cup\QQ\cup\{\ppp\}$ and $\PP^\prime\cup\QQ$ respectively. The weight of all the edges not shown in the figure are within $-1$ to $1$. The weight of all unlabeled edges are $4t$. For simplicity, we do not show edges among $\{u_1, \ldots, u_m\}$.}\label{fig:copelandOM}
 \end{figure}

 We have only one manipulator who tries to make $c$ winner. Now we show that the X3C instance $(\mathcal{U},\mathcal{S})$ is a \YES instance if and only if the \OM instance $(\PP\cup\QQ\cup\{\ppp\},1,c)$ is a \NO instance. Since the number of voters is odd, $\alpha$ does not play any role in the reduction and thus from here on we simply omit $\alpha$. Notice that we can assume without loss of generality that the manipulator's vote prefers $c$ to every other candidate and $x$ to $y$.
 
 In the forward direction, let us now assume that the X3C instance is a \YES instance. Suppose (by renaming) that $S_1, \dots, S_{\nfrac{m}{3}}$ forms an exact set cover. Suppose the manipulator's vote $\vvv$ order $z_1, z_2,$ and $z_3$ as $z_1 \succ z_2\succ z_3$. We will show that $\vvv$ is not a $c$-optimal vote. Symmetrically, we can show that the manipulator's vote ordering $z_1, z_2,$ and $z_3$ in any other order is not $c$-optimal. Consider the following extension $\overline{\PP}$ of $\PP$ and $\overline{\ppp}$ of $\ppp$.
 
 $$ 1\le i\le \nfrac{m}{3}, y \succ c \succ S_i \succ x \succ \text{others} $$
 $$ \nfrac{m}{3} + 1 \le i\le t, S_i \succ x \succ y \succ c \succ \text{others} $$
 $$ \overline{\ppp} = z_1 \succ z_2 \succ z_3 \succ \text{others } $$
 
 From the votes in $\overline{\PP}\cup\QQ\cup\{\vvv, \overline{\ppp}\}$, the Copeland score of $c$ is $m+4$ (defeating $y, z_k, u_j ~\forall k\in[3], j\in[m]$), of $y$ is $m+3$ (defeating $z_k, u_j ~\forall k\in[3], j\in[m]$), of $u_j$ is at most $\nfrac{2m}{3} + 4$ (defeating $x, d_i ~\forall i\in[3]$ and at most $\nfrac{2m}{3}$ many $u_\el\in\UU$), of $x$ is $5$ (defeating $c, y, z_k ~\forall l\in[3]$), of $d_1, d_2$ is $2$ (defeating $y$ and $c$), of $d_3$ is $5$ (defeating $y, c, z_k~\forall k\in[3]$). of $z_3$ is $m+3$ (defeating $d_i, u_j \forall i\in[3], j\in[m]$) for every $k\in[3]$, of $z_3$ is $m+2$ (defeating $d_1, d_2, u_j i\in[3], j\in[m]$), $z_2$ is $m+3$ (defeating $d_1, d_2, z_3, u_j i\in[3], j\in[m]$), $z_1$ is $m+4$ (defeating $d_1, d_2, z_2, z_3, u_j i\in[3], j\in[m]$). Hence, $c$ co-wins with $z_1$ with Copeland score $m+4$. However, the manipulator's vote $c\succ z_3\succ z_2 \succ z_1$ makes $c$ win uniquely. Hence, $\vvv$ is not a $c$-optimal vote and thus the \OM instance is a \NO instance.
 
 For the reverse direction, we show that if the X3C instance is a \NO instance, then there does not exist a vote $\vvv$ of the manipulator and an extension $\overline{\PP}$ of $\PP$ such that $c$ is the unique winner for the votes $\overline{\PP}\cup\QQ\cup\{\vvv^\prime\}$ thereby proving that the \OM instance is vacuously \YES (and thus every vote is $c$-optimal). Consider any extension $\overline{\PP}$ of $\PP$. Notice that, for $c$ to win uniquely, $c$ must defeat each $u_j\in\UU$ and thus $c$ is preferred over $u_j$ in at least one vote in $\overline{\PP}$; we call these votes $\PP_1$. However, in every vote in $\PP_1$, $y$ is preferred over $x$ and thus $|\PP_1|\le \nfrac{m}{3}$ because $x$ must defeat $y$ for $c$ to win uniquely. Since the X3C instance is a \NO instance, there must be a candidate $u\in\UU$ which is not covered by the sets corresponding to the votes in $\PP_1$ and thus $u$ is preferred over $c$ in every vote in $\PP$. Hence, $c$ cannot win uniquely irrespective of the vote of the manipulator. Thus every vote is $c$-optimal and the \OM instance is a \YES instance.
\end{proof}

For the Bucklin and simplified Bucklin voting rules, we show intractability of the \OM problem with at most $15$ undetermined pairs per vote and only one manipulator.

\begin{theorem}\label{thm:bucklinOM}
 The \OM problem is \coNPH for the Bucklin and simplified Bucklin voting rules even when the number of manipulators is one and the number of undetermined pairs in each vote is no more than $15$.
\end{theorem}

\begin{proof}
 We reduce X3C to \OM for the Bucklin and simplified Bucklin voting rules. Let $(\mathcal{U} = \{u_1, \ldots, u_m\}, \mathcal{S}=\{S_1,S_2, \dots, S_t\})$ is an X3C instance. We assume without loss of generality that $m$ is not divisible by $6$ (if not, we introduce three elements in $\UU$ and a set containing them in $\SS$) and $t$ is an even integer (if not, we duplicate any set in $\SS$). We construct a corresponding \OM instance for the Bucklin and simplified Bucklin voting rules as follows.
 $$ \text{Candidate set } \mathcal{C} = \mathcal{U} \cup \{c, z_1, z_2, x_1, x_2, d\} \cup W, \text{ where } |W|=m-3 $$
 For every $i\in[t]$, we define $\mathcal{P}^\prime_i$ as follows:
 
 $$ \forall i\le t, (\UU\setminus S_i) \succ S_i \succ d \succ x_1 \succ x_2 \succ z_1 \succ z_2 \succ \text{ others}$$
 
 Using $\PP^\prime_i$, we define partial vote $\PP_i = \PP^\prime_i \setminus (\{ (\{d\} \cup S_i) \times \{x_1, x_2, z_1, z_2\} \}\cup \{(z_1, z_2)\})$ for every $i\in[t]$. We denote the set of partial votes $\{\PP_i: i\in[t]\}$ by $\PP$ and $\{\PP^\prime_i: i\in[t]\}$ by $\PP^\prime$. We note that the number of undetermined pairs in each partial vote $\PP_i$ is $15$. We introduce the following additional complete votes $\QQ$:
 \begin{itemize}
  \item $\nfrac{t}{2}-\lfloor\nfrac{m}{6}\rfloor-1$ copies of $W\succ z_1\succ z_2\succ x_1\succ c\succ \text{ others}$
  \item $\nfrac{t}{2}-\lfloor\nfrac{m}{6}\rfloor-1$ copies of $W\succ z_1\succ z_2\succ x_2\succ c\succ \text{ others}$
  \item $2\lceil\nfrac{m}{6}\rceil$ copies of $W\succ z_1\succ z_2\succ d\succ c\succ \text{ others}$
  \item $\lfloor\nfrac{m}{6}\rfloor$ copies of $W\succ z_1\succ d\succ x_1\succ c\succ \text{ others}$
  \item $\lfloor\nfrac{m}{6}\rfloor$ copies of $W\succ z_1\succ d\succ x_2\succ c\succ \text{ others}$
  \item $2\lceil\nfrac{m}{6}\rceil-1$ copies of $\UU\succ x_1\succ \text{ others}$
  \item One $\UU\succ c\succ \text{ others}$
 \end{itemize}
 
 We have only one manipulator who tries to make $c$ winner. Now we show that the X3C instance $(\mathcal{U},\mathcal{S})$ is a \YES instance if and only if the \OM instance $(\PP\cup\QQ,1,c)$ is a \NO instance. The total number of voters in the \OM instance is $2t+\nfrac{2m}{3}+1$. We notice that within top $m+1$ positions of the votes in $\PP^\prime\cup\QQ$, $c$ appears $t+\nfrac{m}{3}$ times, $z_1$ and $z_2$ appear $t+\lfloor\nfrac{m}{6}\rfloor$ times, $x_1$ appears $\nfrac{t}{2}+\nfrac{m}{3}-1$ times, $x_2$ appears $\nfrac{t}{2}-1$ times, every candidate in $W$ appears $t+\nfrac{m}{3}-1$ times, every candidate in $\UU$ appears $t+\nfrac{m}{3}+1$ times. Also every candidate in $\UU$ appears $t+\nfrac{m}{3}+1$ times within top $m$ positions of the votes in $\PP\cup\QQ$. Hence, for both Bucklin and simplified Bucklin voting rules, we can assume without loss of generality that the manipulator puts $c$, every candidate in $W$, $x_1$, $x_2$, and exactly one of $z_1$ and $z_2$.
 
 In the forward direction, let us now assume that the X3C instance is a \YES instance. Suppose (by renaming) that $S_1, \dots, S_{\nfrac{m}{3}}$ forms an exact set cover. Suppose the manipulator's vote $\vvv$ puts $c$, every candidate in $W$, $x_1$, $x_2$, and $z_1$ within top $m+1$ positions. We will show that $\vvv$ is not $c$-optimal. The other case where the manipulator's vote $\vvv^\prime$ puts $c$, every candidate in $W$, $x_1$, $x_2$, and $z_2$ within top $m+1$ positions is symmetrical. Consider the following extension $\overline{\PP}$ of $\PP$:
 $$ 1\le i\le \lfloor\nfrac{m}{6}\rfloor, (\UU\setminus S_i) d \succ x_1 \succ x_2 \succ z_2 \succ S_i \succ \succ z_1 \succ \text{ others} $$
 $$ \lceil\nfrac{m}{6}\rceil \le i\le \nfrac{m}{3}, (\UU\setminus S_i) d \succ x_1 \succ x_2 \succ z_1 \succ S_i \succ \succ z_2 \succ \text{ others} $$
 $$ \nfrac{m}{3}+1 \le i\le t, (\UU\setminus S_i) \succ S_i \succ d \succ x_1 \succ x_2 \succ z_1 \succ z_2 \succ \text{ others}$$
 For both Bucklin and simplified Bucklin voting rules, $c$ co-wins with $z_1$ for the votes in $\PP\cup\QQ\cup\{\vvv\}$. However, $c$ wins uniquely for the votes in $\PP\cup\QQ\cup\{\vvv^\prime\}$. Hence, $\vvv$ is not a $c$-optimal vote and thus the \OM instance is a \NO instance. 
 
 For the reverse direction, we show that if the X3C instance is a \NO instance, then there does not exist a vote $\vvv$ of the manipulator and an extension $\overline{\PP}$ of $\PP$ such that $c$ is the unique winner for the votes $\overline{\PP}\cup\QQ\cup\{\vvv^\prime\}$ thereby proving that the \OM instance is vacuously \YES (and thus every vote is $c$-optimal). Consider any extension $\overline{\PP}$ of $\PP$. Notice that, for $c$ to win uniquely, every candidate must be pushed out of top $m+1$ positions in at least one vote in $\PP$; we call these set of votes $\PP_1$. Notice that, $|\PP_1|\ge \nfrac{m}{3}$. However, in every vote in $\PP_1$, at least one of $z_1$ and $z_2$ appears within top $m+1$ many positions. Since, the manipulator has to put at least one of $z_1$ and $z_2$ within its top $m+1$ positions and $z_1$ and $z_2$ appear $t+\lfloor\nfrac{m}{6}\rfloor$ times in the votes in $\PP^\prime\cup\QQ$, we must have $|\PP_1|\le \nfrac{m}{3}$ and thus $|\PP_1|=\nfrac{m}{3}$, for $c$ to win uniquely. However, there exists a candidate $u\in\UU$ not covered by the $S_i$s corresponding to the votes in $\PP_1$. Notice that $u$ gets majority within top $m$ positions of the votes and $c$ can never get majority within top $m+1$ positions of the votes. Hence, $c$ cannot win uniquely irrespective of the vote of the manipulator. Thus every vote is $c$-optimal and the \OM instance is a \YES instance.
\end{proof}

\section{Polynomial Time Algorithms}\label{sec:poly}

We now turn to the polynomial time cases depicted in~\Cref{table:partial_summary}. This section is organized in three parts, one for each problem considered. 

\subsection{Weak Manipulation}
Since the \PW problem is in \Pb{} for the plurality and the veto voting rules~\cite{betzler2009towards}, it follows from \Cref{pw_hard} that the \WM problem is in \Pb{} for the plurality and veto voting rules for any number of manipulators. 

\begin{proposition}\label{pv_easy}
 The \WM problem is in \Pb{} for the plurality and veto voting rules for any number of manipulators.
\end{proposition}

\begin{proof}
 The \PW problem is in \Pb{} for the plurality and the veto voting rules~\cite{betzler2009towards}. Hence, the result follows from \Cref{pw_hard}.
\end{proof}

\subsection{Strong Manipulation}

We now discuss our algorithms for the \SM{} problem. The common flavor in all our algorithms is the following: we try to devise an extension that is as adversarial as possible for the favorite candidate $c$, and if we can make $c$ win in such an extension, then roughly speaking, such a strategy should work for other extensions as well (where the situation only improves for $c$). However, it is challenging to come up with an extension that is globally dominant over all the others in the sense that we just described. So what we do instead is we consider every potential nemesis $w$ who might win instead of $c$, and we build profiles that are ``as good as possible'' for $w$ and ``as bad as possible'' for $c$. Each such profile leads us to constraints on how much the manipulators can afford to favor $w$ (in terms of which positions among the manipulative votes are \textit{safe} for $w$). We then typically show that we can determine whether there exists a set of votes that respects these constraints, either by using a greedy strategy or by an appropriate reduction to a flow problem. We note that the overall spirit here is similar to the approaches commonly used for solving the \NW problem, but as we will see, there are non-trivial differences in the details. We begin with the $k$-approval and $k$-veto voting rules.

\begin{theorem}\label{sm_k_easy}
The \SM problem is in \Pb{} for the $k$-approval and $k$-veto voting rules, for any $k$ and any number of manipulators.
\end{theorem}

\begin{proof}
 For the time being, we just concentrate on non-manipulators' votes. For each candidate $c^{\prime} \in \mathcal{C}\setminus\{c\}$, calculate the maximum possible value of $s^{max}_{NM}(c,c^{\prime}) = s_{NM}(c^{\prime}) - s_{NM}(c)$ from non-manipulators' votes, where $s_{NM}(a)$ is the score that candidate $a$  receives from the votes of the non-manipulators. This can be done by checking all $4$ possible score combinations that $c$ and $c^{\prime}$ can get in each vote $v$ and choosing the one which maximizes $s_v(c^{\prime}) - s_v(c)$ from that vote. We now fix the position of $c$ at the top position for the manipulators' votes and we check if it is possible to place other candidates in the manipulators' votes such that the final value of $s^{max}_{NM}(c,c^{\prime}) + s_M(c^{\prime}) - s_M(c)$ is negative which can be solved easily by reducing it to the max flow problem which is polynomial time solvable.
\end{proof}

We now prove that the \SM problem for scoring rules is in \Pb{} for one manipulator.

\begin{theorem}\label{sm_sr_easy}
The \SM problem is in \Pb{} for any scoring rule when we have only one manipulator.
\end{theorem}

\begin{proof}
For each candidate $c^{\prime} \in \mathcal{C}\setminus\{c\}$, calculate $s^{max}_{NM}(c,c^{\prime})$ using same technique described in the proof of \Cref{sm_k_easy}. We now put $c$ at the top position of the manipulator's vote. For each candidate $c^{\prime} \in \mathcal{C}\setminus\{c\}$, $c^{\prime}$ can be placed at positions $i\in \{2,\ldots,m\}$ in the manipulator's vote which makes $s^{max}_{NM}(c,c^{\prime}) + \alpha_i - \alpha_1$ negative. Using this, construct 
a bipartite graph with $\mathcal{C}\setminus\{c\}$ on left and $\{2, \dots, m\}$ 
on right and there is an edge between $c^{\prime}$ and $i$ iff the candidate $c^{\prime}$ 
can be placed at $i$ in the manipulator's vote according to the above criteria. 
Now solve the problem by finding existence of perfect matching in this graph.
\end{proof}

Our next result proves that the \SM problem for the Bucklin, simplified Bucklin, Fallback, and simplified Fallback voting rules are in \Pb{}.

\begin{theorem}\label{sm_bucklin_easy}
The \SM problem is in \Pb{} for the Bucklin, simplified Bucklin, Fallback, and simplified Fallback voting rules, for any number of manipulators. 
\end{theorem}

\begin{proof}
Let $(\mathcal{C},\mathcal{P},M,c)$ be an instance of \SM for simplified Bucklin, and let $m$ denote the total number of candidates in this instance. Recall that the manipulators have to cast their votes so as to ensure that the candidate $c$ wins in every possible extension of $\mathcal{P}$. We use $\mathcal{Q}$ to denote the set of manipulating votes that we will construct. To begin with, without loss of generality, the manipulators place $c$ in the top position of all their votes. We now have to organize the positioning of the remaining candidates across the votes of the manipulators to ensure that $c$ is a necessary winner of the profile $(\mathcal{P}, \mathcal{Q})$.

To this end, we would like to develop a system of constraints indicating the overall number of times that we are free to place a candidate $w \in \mathcal{C} \setminus \{c\}$ among the top $\ell$ positions in the profile $\mathcal{Q}$. In particular, let us fix $w \in \mathcal{C} \setminus \{c\}$ and $2 \leq \ell \leq m$. Let $\eta_{w,\ell}$ be the maximum number of votes of $\mathcal{Q}$ in which $w$ can appear in the top $\ell$ positions. Our first step is to compute necessary conditions for $\eta_{w,\ell}$.

We use $\overline{\mathcal{P}}_{w,\ell}$ to denote a set of complete votes that we will construct based on the given partial votes. Intuitively, these votes will represent the ``worst'' possible extensions from the point of view of $c$ when pitted against $w$. These votes are engineered to ensure that the manipulators can make $c$ win the elections $\overline{\mathcal{P}}_{w,\ell}$ for all $w \in \mathcal{C} \setminus \{c\}$ and $\ell \in \{2,\ldots,m\}$, if, and only if, they can strongly manipulate in favor of $c$. More formally, there exists a voting profile $\mathcal{Q}$ of the manipulators so that $c$ wins the election $\overline{\mathcal{P}}_{w,\ell} \cup \mathcal{Q}$, for all $w \in \mathcal{C} \setminus \{c\}$ and $\ell \in \{2,\ldots,m\}$ if and only if $c$ wins in every extension of the profile $\mathcal{P} \cup \mathcal{Q}$. 

We now describe the profile $\overline{\mathcal{P}}_{w,\ell}$. The construction is based on the following case analysis, where our goal is to ensure that, to the extent possible, we position $c$ out of the top $\ell-1$ positions, and incorporate $w$ among the top $\ell$ positions. 
\begin{itemize}
 \item Let $v \in \mathcal{P}$ be such that either $c$ and $w$ are incomparable or $w \succ c$. We add the complete vote $v^\prime$ to  $\overline{\mathcal{P}}_{w,\ell}$, where $v^\prime$ is obtained from $v$ by placing $w$ at the highest possible position and $c$ at the lowest possible position, and extending the remaining vote arbitrarily. 
 \item Let $v \in \mathcal{P}$ be such that $c \succ w$, but there are at least $\ell$ candidates that are preferred over $w$ in $v$. We add the complete vote $v^\prime$ to  $\overline{\mathcal{P}}_{w,\ell}$, where $v^\prime$ is obtained from $v$ by placing $c$ at the lowest possible position, and extending the remaining vote arbitrarily.  
 \item Let $v \in \mathcal{P}$ be such that $c$ is forced to be within the top $\ell-1$ positions, then we add the complete vote $v^\prime$ to  $\overline{\mathcal{P}}_{w,\ell}$, where $v^\prime$ is obtained from $v$ by first placing $w$ at the highest possible position followed by placing $c$ at the lowest possible position, and extending the remaining vote arbitrarily.
 \item In the remaining votes, notice that whenever $w$ is in the top $\ell$ positions, $c$ is also in the top $\ell-1$ positions. Let $\mathcal{P}^*_{w,\ell}$ denote this set of votes, and let $t$ be the number of votes in $\mathcal{P}^*_{w,\ell}$.
\end{itemize}
We now consider two cases. Let $d_\ell(c)$ be the number of times $c$ is placed in the top $\ell-1$ positions in the profile $\overline{\mathcal{P}}_{w,\ell} \cup \mathcal{Q}$, and let $d_\ell(w)$ be the number of times $w$ is placed in the top $\ell$ positions in the profile $\overline{\mathcal{P}}_{w,\ell}$. Let us now formulate the requirement that in $\overline{\mathcal{P}}_{w,\ell} \cup \mathcal{Q}$, the candidate $c$ does \emph{not} have a majority in the top $\ell-1$ positions and $w$ \emph{does} have a majority in the top $\ell$ positions. Note that if this requirement holds for any $w$ and $\ell$, then strong manipulation is not possible. Therefore, to strongly manipulate in favor of $c$, we must ensure that for every choice of $w$ and $\ell$, we are able to negate the conditions that we derive. 

The first condition from above simply translates to $d_\ell(c) \le \nfrac{n}{2}$. The second condition amounts to requiring first, that there are at least $\nfrac{n}{2}$ votes where $w$ appears in the top $\ell$ positions, that is, $d_\ell(w) + \eta_{w,\ell} + t > \nfrac{n}{2}$. Further, note that the gap between $d_\ell(w)+\eta_{w,\ell}$ and majority will be filled by using votes from $\mathcal{P}^*_{w,\ell}$ to ``push'' $w$ forward. However, these votes  contribute equally to $w$ and $c$ being in the top $\ell$ and $\ell-1$ positions, respectively. Therefore, the difference between $d_\ell(w)+\eta_{w,\ell}$ and $\nfrac{n}{2}$ must be less than the difference between $d_\ell(c)$ and $\nfrac{n}{2}$. Summarizing, the following conditions, which we collectively denote by $(\star)$, are sufficient to defeat $c$ in some extension: $d_\ell(c) \le \nfrac{n}{2}, d_\ell(w) + \eta_{w,\ell} + t > \nfrac{n}{2}, \nfrac{n}{2} - d_\ell(w) + \eta_{w,\ell} < \nfrac{n}{2} - d_\ell(c)$.

From the manipulator's point of view, the above provides a set of constraints to be satisfied as they place the remaining candidates across their votes. Whenever $d_\ell(c) > \nfrac{n}{2}$, the manipulators place any of the other candidates among the top $\ell$ positions freely, because $c$ already has majority. On the other hand, if $d_\ell(c) \leq \nfrac{n}{2}$, then the manipulators must respect at least one of the following constraints: $\eta_{w,\ell} \leq \nfrac{n}{2} - t - d_\ell(w)$ and $\eta_{w,\ell} \leq d_\ell(c) - d_\ell(w)$.

Extending the votes of the manipulator while respecting these constraints (or concluding that this is impossible to do) can be achieved by a natural greedy strategy --- construct the manipulators' votes by moving positionally from left to right. For each position, consider each manipulator and populate her vote for that position with any available candidate. We output the profile if the process terminates by completing all the votes, otherwise, we say \textsc{No}.

We now argue the proof of correctness. Suppose the algorithm returns \textsc{No}. This implies that there exists a choice of $w \in \mathcal{C} \setminus \{c\}$ and $\ell \in \{2,\ldots,m\}$ such that for any voting profile $\mathcal{Q}$ of the manipulators, the conditions in $(\star)$ are satisfied. (Indeed, if there exists a voting profile that violated at least one of these conditions, then the greedy algorithm would have discovered it.) Therefore, no matter how the manipulators cast their vote, there exists an extension where $c$ is defeated. In particular, for the votes in $\mathcal{P} \setminus \mathcal{P}^*_{w,\ell}$, this extension is given by $\overline{\mathcal{P}}_{w,\ell}$. Further, we choose $\nfrac{n}{2} - \eta_{w,\ell} - d_\ell(w)$ votes among the votes in $\mathcal{P}^*_{w,\ell}$ and extend them by placing $w$ in the top $\ell$ positions (and extending the rest of the profile arbitrary). We extend the remaining votes in $\mathcal{P}^*_{w,\ell}$ by positioning $w$ outside the top $\ell$ positions. Clearly, in this extension, $c$ fails to achieve majority in the top $\ell-1$ positions while $w$ does achieve majority in the top $\ell$ positions. 

On the other hand, if the algorithm returns \textsc{Yes}, then consider the voting profile of the manipulators. We claim that $c$ wins in every extension of $\mathcal{P} \cup \mathcal{Q}$. Suppose, to the contrary, that there exists an extension $\mathcal{R}$ and a candidate $w$ such that the simplified Bucklin score of $c$ is no more than the simplified Bucklin score of $w$ in $\mathcal{R}$. In this extension, therefore, there exists $\ell \in \{2,\ldots,m\}$ for which $w$ attains majority in the top $\ell$ positions and $c$ fails to attain majority in the top $\ell-1$ positions. However, note that this is already impossible in any extension of the profile $\overline{\mathcal{P}}_{w,l} \cup \mathcal{P}^*_{w,\ell}$, because of the design of the constraints. By construction, the number of votes in which $c$ appears in the top $\ell-1$ positions in $\mathcal{R}$ is only greater than the number of times $c$ appears in the top $\ell-1$ positions in any extension of $\overline{\mathcal{P}}_{w,l} \cup \mathcal{P}^*_{w,\ell}$ (and similarly for $w$). This leads us to the desired contradiction. 

For the Bucklin voting rule, we do the following modifications to the algorithm. If $d_\el(c) > d_\el(w)$ for some $w\in\CC\setminus\{c\}$ and $\el<m$, then we make $\eta_{w,\el} = \infty$. The proof of correctness for the Bucklin voting rule is similar to the proof of correctness for the simplified Bucklin voting rule above. 

For Fallback and simplified Fallback voting rules, we consider the number of candidates each voter approves while computing $\eta_{w,\el}$. We output \YES if and only if $\eta_{w,\el}\ge 0$ for every $w\in\CC\setminus\{c\}$ and every $\el\le m$, since we can assume, without loss of generality, that the manipulator approves the candidate $c$ only. Again the proof of correctness is along similar lines to the proof of correctness for the simplified Bucklin voting rule.
\end{proof}

We next show that the \SM problem for the maximin voting rule is polynomial-time solvable when we have only one manipulator.

\begin{theorem}\label{sm_maximin_easy}
The \SM problem for the maximin voting rules are in \Pb{}, when we have only one manipulator.
\end{theorem}

\begin{proof}
 For the time being, just concentrate on non-manipulators' votes. Using the algorithm for NW for maximin in \cite{xia2008determining}, we compute for all pairs $w, w^{\prime} \in \mathcal{C}$, $N_{(w,w^{\prime})}(w,d)$ and $N_{(w,w^{\prime})}(c,w^{\prime})$ for all $d\in \mathcal{C}\setminus \{c\}$. This can be computed in polynomial time. Now we place $c$ at the top position in the manipulator's vote and increase all $N_{(w,w^{\prime})}(c,w^{\prime})$ by one. Now we place a candidate $w$ at the second position if for all $w^{\prime} \in \mathcal{C}$, $N_{(w,w^{\prime})}^{\prime}(w,d) < N_{(w,w^{\prime})}(c,w^{\prime})$ for all $d\in \mathcal{C}\setminus \{c\}$, where $N_{(w,w^{\prime})}^{\prime}(w,d) = N_{(w,w^{\prime})}(w,d)$ of the candidate $d$ has already been assigned some position in the manipulator's vote, and $N_{(w,w^{\prime})}^{\prime}(w,d) = N_{(w,w^{\prime})}(w,d)+1$ else. The correctness argument is in the similar lines of the classical greedy manipulation algorithm of~\cite{bartholdi1989computational}.
\end{proof}

\subsection{Opportunistic Manipulation}

For the plurality, Fallback, and simplified Fallback voting rules, it turns out that the voting profile where all the manipulators approve only $c$ is a $c$-opportunistic voting profile, and therefore it is easy to devise a manipulative vote. 

\begin{observation}\label{obs:plurality_fallback}
 The \OM problem is in \Pb for the plurality and Fallback voting rules for a any number of manipulators.
\end{observation}

For the veto voting rule, however, a more intricate argument is needed, that requires building a system of constraints and a reduction to a suitable instance of the maximum flow problem in a network, to show polynomial time tractability of \OM.

\begin{theorem}\label{thm:vetoOM}
 The \OM problem is in \Pb for the veto voting rule for a constant number of manipulators.
\end{theorem}

\begin{proof}
 Let $(\PP, \el, c)$ be an input instance of \OM. We may assume without loss of generality that the manipulators approve $c$. We view the voting profile of the manipulators as a tuple $(n_a)_{a\in\CC\setminus\{c\}}\in(\mathbb{N}\cup\{0\})^{m-1}$ with $\sum_{a\in\CC\setminus\{c\}} n_a = \el$, where the $n_a$ many manipulators disapprove $a$. We denote the set of such tuples as $\TT$ and we have $|\TT|=O((2m)^\el)$ which is polynomial in $m$ since \el is a constant. A tuple $(n_a)_{a\in\CC\setminus\{c\}}\in\TT$ is not $c$-optimal if there exists another tuple $(n_a^\prm)_{a\in\CC\setminus\{c\}}\in\TT$ and an extension $\overline{\PP}$ of $\PP$ with the following properties. We denote the veto score of a candidate from $\PP$ by $s(\cdot)$. For every candidate $a\in\CC\setminus\{c\}$, we define two quantities $w(a)$ and $d(a)$ as follows.
 \begin{itemize}
  \item $s(c) > s(a)$ for every $a\in\CC\setminus\{c\}$ with $n_a = n_a^\prm = 0$ and we define $w(a) = s(c) - 1, d(a)=0$
  \item $s(c) > s(a) - n_a^\prm  $ for every $a\in\CC\setminus\{c\}$ with $n_a \ge n_a^\prm$ and we define $w(a) = s(c) - n_a^\prm - 1, d(a)=0$
  \item $s(a) - n_a \ge s(c) > s(a) - n_a^\prm$ for every $a\in\CC\setminus\{c\}$ with $n_a < n_a^\prm$ and we define $w(a) = s(c) - n_a^\prm, d(a)=s(a) - n_a$
 \end{itemize}
 We guess the value of $s(c)$. Given a value of $s(c)$, we check the above two conditions by reducing this to a max flow problem instance as follows. We have a source vertex $s$ and a sink $t$. We have a vertex for every $a\in\CC$ (call this set of vertices $Y$) and a vertex for every vote $v\in\PP$ (call this set of vertices $X$). We add an edge from $s$ to each in $X$ of capacity one. We add an edge of capacity one from a vertex $x\in X$ to a vertex $y\in Y$ if the candidate corresponding to the vertex $y$ can be placed at the last position in an extension of the partial vote corresponding to the vertex $x$. We add an edge from a vertex $y$ to $t$ of capacity $w(a)$, where $a$ is the voter corresponding to the vertex $y$. We also set the demand of every vertex $y$ $d(a)$ (that is the total amount of flow coming into vertex $y$ must be at least $d(a)$), where $a$ is the voter corresponding to the vertex $y$. Clearly, the above three conditions are met if and only if there is a feasible $|\PP|$ amount of flow in the above flow graph. Since $s(c)$ can have only $|\PP|+1$ possible values (from $0$ to $\PP$) and $|\TT|=O((2m)^\el)$, we can iterate over all possible pairs of tuples in $\TT$ and all possible values of $s(c)$ and find a $c$-optimal voting profile if there exists a one.
\end{proof}

\section{Conclusion}\label{sec:con}
We revisited many settings where the complexity barrier for manipulation was non-existent, and studied the problem under an incomplete information setting. Our results present a fresh perspective on the use of computational complexity as a barrier to manipulation, particularly in cases that were thought to be dead-ends (because the traditional manipulation problem was polynomially solvable). To resurrect the argument of computational hardness, we have to relax the model of complete information, but we propose that the incomplete information setting is more realistic, and many of our hardness results work even with very limited incompleteness of information. 

Our work is likely to be the starting point for further explorations. To begin with, we leave open the problem of completely establishing the complexity of strong, opportunistic, and weak manipulations for all the scoring rules. Other fundamental forms of manipulation and control do exist in voting, such as destructive manipulation and control by adding candidates. It would be interesting to investigate the complexity of these problems in a partial information setting. 

Another exciting direction is the study of average case complexity, as opposed to the worst case results that we have pursued. These studies have already been carried out in the setting of complete information~\cite{procaccia2006junta,faliszewski2010ai,walsh2010empirical}. Studying the problems that we propose in the average-case model would reveal further insights on the robustness of the incomplete information setting as captured by our model involving partial orders.

Our results showed that the impact of paucity of information on the computational complexity of manipulation crucially depends on the notion of manipulation under consideration. We also argued that different notions of manipulation may be applicable to different situations, maybe based of how optimistic (or pessimistic) the manipulators are. One important direction of future research is to run extensive experimentations on real and synthetic data to know how people manipulate in the absence of complete information.

\subsubsection*{Acknowledgement} Palash Dey wishes to gratefully acknowledge support from Google India for providing him with a special fellowship for carrying out his doctoral work. Neeldhara Misra acknowledges support by the INSPIRE Faculty Scheme, DST India (project IFA12-ENG-31).

\bibliographystyle{alpha}
\bibliography{refThesis}

\newcommand{\etalchar}[1]{$^{#1}$}
\begin{thebibliography}{DKNW11}

\bibitem[BBF10]{bachrach2010probabilistic}
Yoram Bachrach, Nadja Betzler, and Piotr Faliszewski.
\newblock Probabilistic possible winner determination.
\newblock In {\em International Conference on Artificial Intelligence (AAAI)},
  volume~10, pages 697--702, 2010.

\bibitem[BCE{\etalchar{+}}15]{brandt2015handbook}
Felix Brandt, Vincent Conitzer, Ulle Endriss, J{\'e}r{\^o}me Lang, and Ariel~D
  Procaccia.
\newblock Handbook of computational social choice, 2015.

\bibitem[BD09]{betzler2009towards}
Nadja Betzler and Britta Dorn.
\newblock Towards a dichotomy of finding possible winners in elections based on
  scoring rules.
\newblock In {\em Mathematical Foundations of Computer Science (MFCS)}, pages
  124--136. Springer, 2009.

\bibitem[BFLR12]{BaumeisterFLR12}
Dorothea Baumeister, Piotr Faliszewski, J{\'{e}}r{\^{o}}me Lang, and J{\"{o}}rg
  Rothe.
\newblock Campaigns for lazy voters: truncated ballots.
\newblock In {\em International Conference on Autonomous Agents and Multiagent
  Systems, {AAMAS} 2012, Valencia, Spain, June 4-8, 2012 {(3} Volumes)}, pages
  577--584, 2012.

\bibitem[BIO91]{bartholdi1991single}
John Bartholdi~III and James~B. Orlin.
\newblock Single transferable vote resists strategic voting.
\newblock {\em Soc. Choice Welf.}, 8(4):341--354, 1991.

\bibitem[BITT89]{bartholdi1989computational}
John Bartholdi~III, C.A. Tovey, and M.A. Trick.
\newblock The computational difficulty of manipulating an election.
\newblock {\em Soc. Choice Welf.}, 6(3):227--241, 1989.

\bibitem[BNW11]{betzler2011unweighted}
Nadja Betzler, Rolf Niedermeier, and Gerhard~J Woeginger.
\newblock Unweighted coalitional manipulation under the borda rule is
  {NP}-hard.
\newblock In {\em IJCAI}, volume~11, pages 55--60, 2011.

\bibitem[BRR11]{baumeister2011computational}
Dorothea Baumeister, Magnus Roos, and J{\"o}rg Rothe.
\newblock Computational complexity of two variants of the possible winner
  problem.
\newblock In {\em The 10th International Conference on Autonomous Agents and
  Multiagent Systems (AAMAS)}, pages 853--860, 2011.

\bibitem[BRR{\etalchar{+}}12]{baumeister2012possible}
Dorothea Baumeister, Magnus Roos, J{\"o}rg Rothe, Lena Schend, and Lirong Xia.
\newblock The possible winner problem with uncertain weights.
\newblock In {\em ECAI}, pages 133--138, 2012.

\bibitem[BS09]{brams2009voting}
Steven~J Brams and M~Remzi Sanver.
\newblock Voting systems that combine approval and preference.
\newblock In {\em The mathematics of preference, choice and order}, pages
  215--237. Springer, 2009.

\bibitem[CLMM10]{chevaleyre2010possible}
Yann Chevaleyre, J{\'e}r{\^o}me Lang, Nicolas Maudet, and J{\'e}r{\^o}me
  Monnot.
\newblock Possible winners when new candidates are added: The case of scoring
  rules.
\newblock In {\em Proc. International Conference on Artificial Intelligence
  (AAAI)}, 2010.

\bibitem[CSL07]{conitzer2007elections}
Vincent Conitzer, Tuomas Sandholm, and J{\'e}r{\^o}me Lang.
\newblock When are elections with few candidates hard to manipulate?
\newblock {\em J. ACM}, 54(3):14, 2007.

\bibitem[CWX11]{conitzer2011dominating}
Vincent Conitzer, Toby Walsh, and Lirong Xia.
\newblock Dominating manipulations in voting with partial information.
\newblock In {\em International Conference on Artificial Intelligence (AAAI)},
  volume~11, pages 638--643, 2011.

\bibitem[Dey15]{dey2015computational}
Palash Dey.
\newblock Computational complexity of fundamental problems in social choice
  theory.
\newblock In {\em Proc. 2015 International Conference on Autonomous Agents and
  Multiagent Systems}, pages 1973--1974. International Foundation for
  Autonomous Agents and Multiagent Systems, 2015.

\bibitem[DKNW11]{davies2011complexity}
Jessica Davies, George Katsirelos, Nina Narodytska, and Toby Walsh.
\newblock Complexity of and algorithms for borda manipulation.
\newblock In {\em Proc. International Conference on Artificial Intelligence
  (AAAI)}, pages 657--662, 2011.

\bibitem[DL13]{ding2013voting}
Ning Ding and Fangzhen Lin.
\newblock Voting with partial information: what questions to ask?
\newblock In {\em Proc. 12th International Conference on Autonomous Agents and
  Multi-agent Systems (AAMAS)}, pages 1237--1238. International Foundation for
  Autonomous Agents and Multiagent Systems, 2013.

\bibitem[DMN15a]{DeyMN15a}
Palash Dey, Neeldhara Misra, and Y.~Narahari.
\newblock Detecting possible manipulators in elections.
\newblock In {\em Proc. 2015 International Conference on Autonomous Agents and
  Multiagent Systems, {AAMAS} 2015, Istanbul, Turkey, May 4-8, 2015}, pages
  1441--1450, 2015.

\bibitem[DMN15b]{DeyMN15}
Palash Dey, Neeldhara Misra, and Y.~Narahari.
\newblock Kernelization complexity of possible winner and coalitional
  manipulation problems in voting.
\newblock In {\em Proc. 2015 International Conference on Autonomous Agents and
  Multiagent Systems, {AAMAS} 2015, Istanbul, Turkey, May 4-8, 2015}, pages
  87--96, 2015.

\bibitem[DMN16]{journalsDeyMN16}
Palash Dey, Neeldhara Misra, and Y.~Narahari.
\newblock Kernelization complexity of possible winner and coalitional
  manipulation problems in voting.
\newblock {\em Theor. Comput. Sci.}, 616:111--125, 2016.

\bibitem[DN14]{dey2014asymptotic}
Palash Dey and Y~Narahari.
\newblock Asymptotic collusion-proofness of voting rules: the case of large
  number of candidates.
\newblock In {\em Proc. 13th International Conference on Autonomous Agents and
  Multiagent Systems (AAMAS)}, pages 1419--1420. International Foundation for
  Autonomous Agents and Multiagent Systems, 2014.

\bibitem[DN15]{dey2015asymptoticjournal}
Palash Dey and Y~Narahari.
\newblock Asymptotic collusion-proofness of voting rules: The case of large
  number of candidates.
\newblock {\em Studies in Microeconomics}, 3(2):120--139, 2015.

\bibitem[EE12]{elkind2012manipulation}
Edith Elkind and G{\'a}bor Erd{\'e}lyi.
\newblock Manipulation under voting rule uncertainty.
\newblock In {\em Proc. 11th International Conference on Autonomous Agents and
  Multiagent Systems (AAMAS)}, pages 627--634. International Foundation for
  Autonomous Agents and Multiagent Systems, 2012.

\bibitem[ER91]{ephrati1991clarke}
Eithan Ephrati and Jeffrey~S Rosenschein.
\newblock The {C}larke tax as a consensus mechanism among automated agents.
\newblock In {\em Proc. Ninth International Conference on Artificial
  Intelligence (AAAI)}, pages 173--178, 1991.

\bibitem[FHH10]{faliszewski2010using}
Piotr Faliszewski, Edith Hemaspaandra, and Lane~A Hemaspaandra.
\newblock Using complexity to protect elections.
\newblock {\em Commun ACM}, 53(11):74--82, 2010.

\bibitem[FHHR09]{FaliszewskiHHR09}
Piotr Faliszewski, Edith Hemaspaandra, Lane~A. Hemaspaandra, and J{\"{o}}rg
  Rothe.
\newblock Llull and copeland voting computationally resist bribery and
  constructive control.
\newblock {\em J. Artif. Intell. Res.}, 35:275--341, 2009.

\bibitem[FHS08]{faliszewski2008copeland}
Piotr Faliszewski, Edith Hemaspaandra, and Henning Schnoor.
\newblock Copeland voting: Ties matter.
\newblock In {\em Proc. 7th International Conference on Autonomous Agents and
  Multiagent Systems (AAMAS)}, pages 983--990. International Foundation for
  Autonomous Agents and Multiagent Systems, 2008.

\bibitem[FHS10]{faliszewski2010manipulation}
Piotr Faliszewski, Edith Hemaspaandra, and Henning Schnoor.
\newblock Manipulation of copeland elections.
\newblock In {\em Proc. 9th International Conference on Autonomous Agents and
  Multiagent Systems (AAMAS)}, pages 367--374. International Foundation for
  Autonomous Agents and Multiagent Systems, 2010.

\bibitem[FKN08]{friedgut2008elections}
Ehud Friedgut, Gil Kalai, and Noam Nisan.
\newblock Elections can be manipulated often.
\newblock In {\em IEEE 49th Annual IEEE Symposium on Foundations of Computer
  Science (FOCS)}, pages 243--249. IEEE, 2008.

\bibitem[FKS03]{Fagin}
Ronald Fagin, Ravi Kumar, and D.~Sivakumar.
\newblock Efficient similarity search and classification via rank aggregation.
\newblock In {\em Proc. 2003 ACM SIGMOD International Conference on Management
  of Data}, SIGMOD '03, pages 301--312, New York, NY, USA, 2003. ACM.

\bibitem[FP10]{faliszewski2010ai}
Piotr Faliszewski and Ariel~D Procaccia.
\newblock Ai's war on manipulation: Are we winning?
\newblock {\em AI Magazine}, 31(4):53--64, 2010.

\bibitem[Gib73]{gibbard1973manipulation}
Allan Gibbard.
\newblock Manipulation of voting schemes: a general result.
\newblock {\em Econometrica}, pages 587--601, 1973.

\bibitem[GJ79]{garey1979computers}
Michael~R Garey and David~S Johnson.
\newblock {\em Computers and {I}ntractability}, volume 174.
\newblock freeman New York, 1979.

\bibitem[GNNW14]{gaspers2014possible}
Serge Gaspers, Victor Naroditskiy, Nina Narodytska, and Toby Walsh.
\newblock Possible and necessary winner problem in social polls.
\newblock In {\em Proc. 13th International Conference on Autonomous Agents and
  Multiagent Systems (AAMAS)}, pages 613--620. International Foundation for
  Autonomous Agents and Multiagent Systems, 2014.

\bibitem[IKM12]{isaksson2012geometry}
M.~Isaksson, G.~Kindler, and E.~Mossel.
\newblock The geometry of manipulation - a quantitative proof of the
  gibbard-satterthwaite theorem.
\newblock {\em Combinatorica}, 32(2):221--250, 2012.

\bibitem[KL05]{konczak2005voting}
Kathrin Konczak and J{\'e}r{\^o}me Lang.
\newblock Voting procedures with incomplete preferences.
\newblock In {\em Proc. International Joint Conference on Artificial
  Intelligence-05 Multidisciplinary Workshop on Advances in Preference
  Handling}, volume~20, 2005.

\bibitem[McG53]{mcgarvey1953theorem}
David~C McGarvey.
\newblock A theorem on the construction of voting paradoxes.
\newblock {\em Econometrica}, pages 608--610, 1953.

\bibitem[ML15]{MenonL15}
Vijay Menon and Kate Larson.
\newblock Complexity of manipulation in elections with partial votes.
\newblock {\em CoRR}, abs/1505.05900, 2015.

\bibitem[NW14]{NarodytskaW14}
Nina Narodytska and Toby Walsh.
\newblock The computational impact of partial votes on strategic voting.
\newblock In {\em Proc. 21st European Conference on Artificial Intelligence,
  18-22 August 2014, Prague, Czech Republic - Including Prestigious
  Applications of Intelligent Systems {(PAIS} 2014)}, pages 657--662, 2014.

\bibitem[PHG00]{PennockHG00}
David~M. Pennock, Eric Horvitz, and C.~Lee Giles.
\newblock Social choice theory and recommender systems: Analysis of the
  axiomatic foundations of collaborative filtering.
\newblock In {\em Proc. Seventeenth National Conference on Artificial
  Intelligence and Twelfth Conference on on Innovative Applications of
  Artificial Intelligence, July 30 - August 3, 2000, Austin, Texas, {USA.}},
  pages 729--734, 2000.

\bibitem[PR06]{procaccia2006junta}
Ariel~D Procaccia and Jeffrey~S Rosenschein.
\newblock Junta distributions and the average-case complexity of manipulating
  elections.
\newblock In {\em Proc. Fifth International Conference on Autonomous Agents and
  Multiagent Systems (AAMAS)}, pages 497--504. ACM, 2006.

\bibitem[PR07]{ProcacciaR07}
Ariel~D. Procaccia and Jeffrey~S. Rosenschein.
\newblock Average-case tractability of manipulation in voting via the fraction
  of manipulators.
\newblock In {\em Proc. 6th International Joint Conference on Autonomous Agents
  and Multiagent Systems {(AAMAS} 2007), Honolulu, Hawaii, USA, May 14-18,
  2007}, page 105, 2007.

\bibitem[Sat75]{satterthwaite1975strategy}
Mark~Allen Satterthwaite.
\newblock Strategy-proofness and {Arrow's} conditions: Existence and
  correspondence theorems for voting procedures and social welfare functions.
\newblock {\em J. Econ. Theory}, 10(2):187--217, 1975.

\bibitem[Wal10]{walsh2010empirical}
Toby Walsh.
\newblock An empirical study of the manipulability of single transferable
  voting.
\newblock In {\em Proc. 19th European Conference on Artificial Intelligence
  (ECAI)}, pages 257--262, 2010.

\bibitem[Wal11]{walsh2011hard}
Toby Walsh.
\newblock Where are the hard manipulation problems?
\newblock {\em J. Artif. Intell. Res.}, pages 1--29, 2011.

\bibitem[XC08a]{xia2008generalized}
Lirong Xia and Vincent Conitzer.
\newblock Generalized scoring rules and the frequency of coalitional
  manipulability.
\newblock In {\em Proc. 9th ACM conference on Electronic Commerce (EC)}, pages
  109--118. ACM, 2008.

\bibitem[XC08b]{xia2008sufficient}
Lirong Xia and Vincent Conitzer.
\newblock A sufficient condition for voting rules to be frequently manipulable.
\newblock In {\em Proc. 9th ACM conference on Electronic Commerce (EC)}, pages
  99--108. ACM, 2008.

\bibitem[XC11]{xia2008determining}
Lirong Xia and Vincent Conitzer.
\newblock Determining possible and necessary winners under common voting rules
  given partial orders.
\newblock volume~41, pages 25--67. AI Access Foundation, 2011.

\bibitem[XZP{\etalchar{+}}09]{xia2009complexity}
Lirong Xia, Michael Zuckerman, Ariel~D Procaccia, Vincent Conitzer, and
  Jeffrey~S Rosenschein.
\newblock Complexity of unweighted coalitional manipulation under some common
  voting rules.
\newblock In {\em Proc. 21st International Joint Conference on Artificial
  Intelligence (IJCAI)}, volume~9, pages 348--352, 2009.

\end{thebibliography}

\end{document}